\documentclass[10pt,journal]{IEEEtran}

\usepackage{lineno}
\usepackage{hyperref}
\modulolinenumbers[5]

\usepackage{graphicx} 
\usepackage{epsfig} 
\usepackage{amsmath} 
\usepackage{amssymb}  
\usepackage{epstopdf}
\usepackage{mathtools}
\usepackage{multirow}
\usepackage{array}
\usepackage{siunitx}
\usepackage{tikz}
\usepackage{mathrsfs}
\usepackage{tabularx}
\usepackage{amsthm}
\usepackage{url}
\usetikzlibrary{calc}
\usetikzlibrary{positioning}
\usetikzlibrary{arrows}
\usetikzlibrary{shapes}
\usetikzlibrary{fit}
\newtheorem{rem}{Remark}
\newtheorem{assm}{Assumption}
\newtheorem{alg}{Algorithm}
\newtheorem{theorem}{Theorem}
\newtheorem{lemma}{Lemma}

\usepackage{algpseudocode}
\usepackage{float}
\newcommand{\pv}[1]{\hat #1}
\newcommand{\ltvpv}[1]{\tilde #1}
\newcommand{\ct}[1]{\grave #1}
\newcommand{\ctscp}[1]{\bar #1}
\newcommand{\boost}{\textrm{boost}}
\newcommand{\EGR}{\textrm{EGR}}
\newcommand{\steady}[1]{\bar #1}
\newcolumntype{L}[1]{>{\raggedright\let\newline\\\arraybackslash\hspace{0pt}}m{#1}}
\newcolumntype{C}[1]{>{\centering\let\newline\\\arraybackslash\hspace{0pt}}m{#1}}
\newcolumntype{R}[1]{>{\raggedleft\let\newline\\\arraybackslash\hspace{0pt}}m{#1}}

\usepackage{psfrag}

\DeclareMathOperator{\diag}{diag}

\begin{document}

	\def\pimbar{200}
	\def\wcbar{50}
	\def\egrbar{100}

	
	\title{Fast Calibration of a Robust Model Predictive Controller for Diesel Engine Airpath}

	\author{Gokul S.~Sankar$^{1}$, Rohan C.~Shekhar$^{1}$, Chris Manzie$^{2}$, Takeshi Sano$^{3}$ and Hayato Nakada$^{3}$
		\thanks{*This work was supported by Australian Research Council (ARC) [Grant number - LP160100650], Toyota Motor Corporation, Japan and The University of Melbourne through Melbourne International Research Scholarship (MIRS) Melbourne-India Postgraduate Program (MIPP).}
		\thanks{$^{1}$Gokul~S.~Sankar and Rohan C.~Shekhar are with the Department of Mechanical Engineering, The University of Melbourne, Victoria 3010, Australia. Emails: {\tt\small ggowri@student.unimelb.edu.au} (G.~Sankar) and 
			{\tt\small rshekhar@unimelb.edu.au} (R.~Shekhar). }%
		\thanks{$^2$ Chris Manzie is with the Department of Electrical and Electronic Engineering, The University of Melbourne, Victoria 3010, Australia. Email: {\tt\small manziec@unimelb.edu.au}.}
		\thanks{$^{3}$Takeshi Sano and Hayato Nakada are with the Advanced Unit Management System Development 		Division, Toyota Motor Corporation,
			Higashi-Fuji Technical Center,
			1200, Mishuku, Susono-city, Shizuoka, 410-1193 Japan. Email:
			{\tt\small \{takeshi\_sano\_aa, hayato\_nakada\}@mail.toyota.co.jp} (T.~Sano, H.~ Nakada).}%
	}

	
	\maketitle
	\thispagestyle{empty}
	\pagestyle{empty}
	
	\begin{abstract}
		
		A significant challenge in the development of control systems for diesel airpath applications is to tune the controller parameters to achieve satisfactory output performance, especially whilst adhering to input and safety constraints in the presence of unknown system disturbances. Model-based control techniques, such as model predictive control (MPC), have been successfully applied to multivariable and highly nonlinear systems, such as diesel engines, while considering operational constraints. However, efficient calibration of typical implementations of MPC is hindered by the high number of tuning parameters and their non-intuitive correlation with the output response. In this paper, the number of effective tuning parameters is reduced through suitable structural modifications to the controller formulation and an appropriate redesign of the MPC cost function to aid rapid calibration. Furthermore, a constraint tightening-like approach is augmented to the control architecture to provide robustness guarantees in the face of uncertainties. A switched linear time-invariant MPC strategy with recursive feasibility guarantees during controller switching is proposed to handle transient operation of the engine. The robust controller is first implemented on a high fidelity simulation environment, with a comprehensive investigation of its calibration to achieve desired transient response under step changes in the fuelling rate. An experimental study then validates and highlights the performance of the proposed controller architecture for selected tunings of the calibration parameters for fuelling steps and over drive cycles.
	\end{abstract}
	
	
	
	
	\section{Introduction}
	
	Diesel engines used in automotive applications must meet increasingly stringent emissions standards whilst also satisfying drivability and fuel efficiency requirements. The actuators in the diesel airpath, namely the throttle valve, the exhaust gas recirculation (EGR) valve and the variable geometry turbocharger (VGT), influence the flows of fresh air and exhaust gas into the engine and thus can be used to control the responsiveness to driver demands and the formations of particulate matter (PM) and Nitrogen oxides $\left(\textrm{NO}_{\textrm{x}}\right)$ emissions. A diesel airpath controller is required to track two reference levels, namely an intake manifold pressure (also known as the boost pressure) reference and a EGR rate reference, which is defined as the ratio of the EGR outflow rate to the combined EGR and compressor outflow rates. For a given engine operating condition, characterised by an engine rotational speed, $\omega_e$, and a fuelling rate, $\dot{m}_f$, the references for boost pressure and EGR rate are determined by a high level controller to provide `optimal' driver demand responsiveness and fuel efficiency, and satisfy emission regulations.
	
	Diesel engine airpath control is challenging owing to its nonlinear multivariable nature \cite{Kolmanovsky1999, Wahlstrom2011}. Conventional control approaches use look-up tables and single-input single-output (SISO) proportional-integral-differential (PID) loops ignoring the cross-sensitivities that can affect the controller performance. Sophisticated multivariable control algorithms will lead to an improved performance \cite{VanNieuwstadt1998}, whilst satisfying operating constraints on manifold pressures and physical limitations of the actuators.
	
	The constraint handling ability of model predictive control (MPC) makes it an ideal choice of control architecture for constrained multi-input multi-output systems, such as diesel engines. However, several factors, namely the computational complexity, immense calibration effort and the plant-model mismatch, impede the implementation of  MPC schemes for diesel engine applications. In the upcoming paragraphs of this paper, the aforementioned challenges, and the approaches proposed and adopted to enable real time implementation of an MPC architecture on a diesel engine are discussed.
	
	Improvements in the control performance obtained by using nonlinear MPC over conventional control approaches for diesel airpath have been illustrated through simulation \cite{Herceg2006}. However, due to the lack of computational resources capable of solving nonlinear MPC problems online within the fast sampling time periods typically used in engine control, linear MPC formulations are applied for diesel airpath \cite{Rueckert2004}. Furthermore, in order to avoid solving the quadratic program (QP) online, explicit MPC formulations are often employed \cite{Ferreau2007,Stewart2008,Ortner2006,Karlsson2010,Huang2016}, where the explicit piecewise-affine control law corresponding to the solution of the parametric QP is computed offline. A further reduction in computational complexity is achieved by strategies such as intermittent constraint enforcement \cite{Huang2013}. However, this results in loss of constraint satisfaction guarantees during the intermediate time steps in which the constraints were not imposed. Similarly, enforcing soft constraints to ensure controller feasibility \cite{Karlsson2010,Wahlstroem2013,Huang2013}, leads to loss of guarantees on constraint adherence. Nevertheless, the growth in the number of regions in explicit MPC restricts the control horizon length to typically one or two steps for real-time implementation \cite{Ortner2007, Huang2016}. With such short control horizon, the controller has fewer degrees of freedom within the prediction horizon, which might lead to aggressive control actions being taken to achieve the desired objective, leading to increased wear and tear of the system. More importantly, these studies have not addressed the issue of requiring high calibration effort for the diesel airpath controller to achieve the desired output transient response. 
	
	The high number of tuning parameters and the non-intuitive relationship between the tuning parameters and the time domain characteristics of the output response, such as overshoot and settling time, are the reasons for the difficulty involved in the calibration of an MPC. In a standard MPC quadratic cost function with symmetric weighting matrices, there are $\frac{n}{2}({n+1}) + \frac{m}{2}({m+1})$ independent cost function parameters that require tuning (excluding the horizon length, terminal cost and set), where $n$ and $m$ are the state and input dimensions, respectively.
	
	In \cite{Rowe2000}, $H_{\infty}$ loopshaping was proposed to tune infinite horizon MPC parameters but requires the system to operate away from the constraint boundary. This naturally cedes some of the advantages of explicitly considering constraints, particularly when finite horizons must be considered for computational reasons. As an alternative, in \cite{Sankar2015},  \cite{Shekhar2017}, a novel parameterisation of the MPC cost function was proposed to provide an explicit relationship between a small number of tuning parameters and the time domain response. This approach is well suited for rapid calibration, and does not require calibrators to have explicit knowledge of advanced control techniques.

	A perfect prediction model is assumed to be available in many previous works on the application of MPC to diesel engines. However, the presence of external disturbances and modelling errors might potentially lead to constraint violation. Robust MPC techniques, such as tube MPC and constraint tightening, can be used to guarantee constraint adherence and recursive feasibility of the controller in the presence of disturbances. In \cite{Huang2014}, a single step horizon is used to ensure computational tractability of the tube MPC for diesel air path, failing to utilise the potential of MPC to anticipate system behaviour over multiple time-steps. However, the constraint tightening (CT) approach, in which the constraints are systematically tightened and a margin is reserved along the horizon to correct for errors due to uncertainties \cite{Richards2006}, does not incur additional computational load when compared to the conventional MPC. In \cite{Broomhead2017}, a CT-MPC formulation is proposed for diesel generators in power tracking applications. In both \cite{Huang2014} and \cite{Broomhead2017}, the number of tuning parameters remain the same as in conventional MPC, requiring a significant effort to achieve the desired output time domain characteristics. Furthermore, as in \cite{Sankar2017}, it might be infeasible to provide robustness guarantees with a large disturbance set because of lack of control authority to reject the disturbances by the end of the horizon in robust MPC formulations. 
	
	Many previous implementations of MPC on diesel engines \cite{Ortner2006,Ortner2007,DelRe2010} divide the engine operational space into multiple regions and identify local models for employing multiple linear predictive controllers. In a transient operation, the controllers are switched  based on the current engine operating point. This switched linear time-invariant (LTI) MPC architecture reduces computational complexity while retaining some desirable attributes of nonlinear MPC. However, switching between controllers might lead to bumps in the output response or loss of recursive feasibility during switching. In order to avoid bumps, \cite{Ortner2006, Ortner2007} use the actual output values from the previous time step for all their local controllers, however recursive feasibility and hence, stability guarantees in the switched LTI-MPC architecture are not provided. The existing CT formulations for such architectures, either assume the future system representations are known \cite{Richards2005}, or do not account for the error due to change in trimming conditions of the linearisation models \cite{Broomhead2014}, thereby, losing robust feasibility guarantees.
	
	This paper is a significant extension of the previous work \cite{Shekhar2017, Sankar2017}. In this paper, a switched LTI-MPC architecture with robust feasibility guarantees using CT approach is proposed to handle transient engine operation. Furthermore, a methodology based on sequential convex programming (SCP) approach \cite{Dinh2010} is proposed to produce less conservative disturbance set estimates that can be handled by the local controllers in comparison with \cite{Sankar2017}. The switched LTI-MPC architecture consisting of local robust controllers is implemented on an engine and shown to maintain all input and state constraints in the presence of allowable disturbances. As a further extension to the ideas proposed in \cite{Shekhar2017}, a procedure is proposed for identifying the switched controllers requiring further tuning using feedback from performance over drive cycles.

	\subsection{Notation}
	The symbol $\mathbb{R}$ represents a set of real numbers. The symbol $\mathbb{Z}_{\left[a: b\right]}$ denotes a set of consecutive integers from $a$ to $b$ and $2\mathbb{Z}^{+}$ denotes set of positive even integers. $0_{m\times n}$ represents a zero matrix of size $m\times n$, $I_n$ denotes an $n\times n$ identity matrix and $\boldsymbol{1}_n$ is a $n \times 1$ vector of ones. The operator $\det\left(A\right)$ denotes the determinant of the matrix $A$. $A \succ 0$ represents a positive definite matrix $A$. The Euclidean norm of a vector $x$ is denoted by $\|x\|$; $\|x\|_1$ represents its $L_1$ norm; and $\|A\|_{\max}\coloneqq\underset{i,j}{\max}\,\left|a_{ij}\right|$, where $\left|a_{ij}\right|$ is the absolute value of the element in $i^{th}$ row and $j^{th}$ column of the matrix $A$. The operator $\ominus$ denotes the Pontryagin difference, defined for the sets $\mathcal{A}$ and $\mathcal{B}$ as
	\begin{equation}
	\mathcal{A}\ominus\mathcal{B}\coloneqq\left\{ a|a+b\in\mathcal{A}\,\forall b\in\mathcal{B}\right\}.\label{eq:pontdiff}
	\end{equation}
	
	The operator $\diag\{\cdot\}$ denotes a diagonal matrix with the elements in parentheses along the leading diagonal. For a sequence of matrices with same row dimension $\left\{A\left(a\right),\,A\left(a+1\right),\,\ldots,\,A\left(b\right)\right\}$, the operator $\mathscr{C}$ concatenates the matrices horizontally as
	\begin{equation}
	\mathscr{C}_{i=a}^b A\left(i\right) = \left\{A\left(a\right),\,A\left(a+1\right),\,\ldots,\,A\left(b\right)\right\}
	\end{equation}
	\noindent where $a,\,b\in \mathbb{Z}$ and $a\leq b$. All inequalities involving vectors are to be interpreted row-wise. 	
		\begin{figure}
		\begin{centering}
			\begin{tikzpicture}
			\node at (0,0)
			{\includegraphics[clip, trim = {0.cm 0.0cm 0.cm 0.cm}]{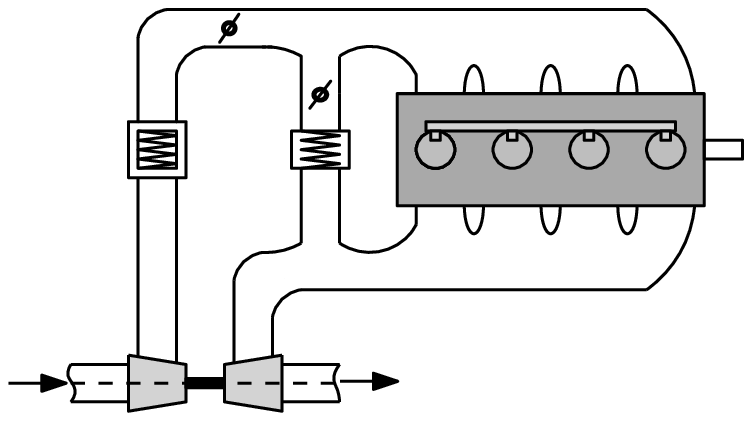}};
			\node at (-3.33,0.65) {Intercooler};
			\node at (-1.33,0.65) {EGR};
			\node at (-1.33,0.375) {cooler};
			\node at (-1.25,1.35) {EGR};
			\node at (-1.25,1.075) {valve};
			\node at (-3.5,-1.4) {Air};
			\node at (-1.5,2.25) {Throttle};
			\node [text width = 2cm, align = center] at (0.6,-1.3)  {Engine-out exhaust gas};
			\node at (1.5,-0.55) {Exhaust manifold};
			\node at (1.5,1.75) {Intake manifold};
			\node at (1.75,0.25) {Cylinders};
			\node at (-2.75,-2.25) {Compressor};
			\node at (-1.2,-2.25) {VGT};
			\node at (1.77,1.025) {\tiny Fuel rail and injectors};
			\end{tikzpicture}
			\par\end{centering}
		\protect\caption{Diesel engine schematic.}
		\label{fig:eng_sch}
	\end{figure}

	\section{System Identification and Disturbance Set Estimation}
	\label{sec:modelling}	
	A schematic representation of the diesel airpath with the positioning of the actuators and other components is shown in Fig.~\ref{fig:eng_sch}. The density of the fresh air entering the airpath is first increased by the compressor and then by the intercooler. A part of the burnt gas in the exhaust manifold is cooled and recirculated to the cylinders through the EGR system. The engine-out exhaust gas drives the VGT, whose shaft spins the compressor. 
	
	As the diesel airpath is a highly nonlinear system, the engine operating range is divided into $12$ regions as shown in Fig.~\ref{fig:mod_grid}, with a linear model representing the engine dynamics in each region. The grid of the selected operating points is evenly spaced in the engine operational range as seen in Fig.~\ref{fig:mod_grid}.
	
	Four states are selected for these models,  as listed in Table \ref{tab:states}. This choice of the states rules out the need for designing a state estimator as complete state feedback is available either through direct sensor measurements or is reliably estimated within the engine control unit (ECU). It is, however, possible to include more states to improve the accuracy of model predictions at the cost of incurring additional online computational load and requiring state estimation.
		\begin{figure}
		\begin{centering}
			\begin{tikzpicture} [scale = 1, transform shape]
			\node at (0,0) (pic) {\includegraphics[clip, trim = {0.65cm 0.55cm 0cm 0.0cm}]{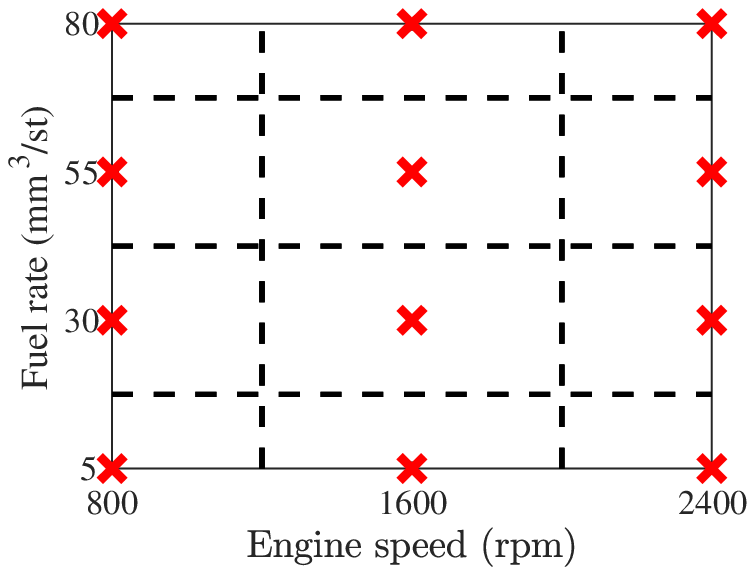}};
			\node [above left = -50mm and -16mm of pic]{I};
			\node [above left = -38.5mm and -16.5mm of pic]{II};
			\node [above left = -23.5mm and -17mm of pic]{III};
			\node [above left = -11.5mm and -17mm of pic]{IV};
			
			\node [above left = -50mm and -47mm of pic]{V};
			\node [above left = -38.5mm and -47.3mm of pic]{VI};
			\node [above left = -23.5mm and -47.8mm of pic]{VII};
			\node [above left = -11.5mm and -48.2mm of pic]{VIII};
			
			\node [above left = -50mm and -64mm of pic]{IX};
			\node [above left = -38.5mm and -63.5mm of pic]{X};
			\node [above left = -23.5mm and -64mm of pic]{XI};
			\node [above left = -11.5mm and -64.5mm of pic]{XII};
			
			
			\node (rect) [draw =none, fill = white, minimum width = 0.65cm, minimum height = 4.95cm, inner sep = 0pt, below  left = -54mm and -4.75mm of pic] {};
			
			\node [ below  left = -9mm and -6mm of pic] {Low};
			\node [ below  left = -24mm and -6mm of pic] {Mid-low};
			\node [ below  left = -39mm and -6mm of pic] {Mid-high};
			\node [ below  left = -54mm and -6mm of pic] {High};

			
			\node (rect) [draw = none, fill = white, minimum width = 6.75cm, minimum height = 0.3cm, inner sep = 0pt, below  left = -4mm and -71mm of pic] {};
			
			\node [ below  left = -5mm and -11mm of pic] {Low};
			\node [ below  left = -5mm and -41mm of pic] {Mid};
			\node [ below  left = -5mm and -72mm of pic] {High};
			
			\node [above left = -13mm and 8mm of pic, rotate = 90]{Fuel rate $\left(\SI{}{mm^3/st}\right)$};
			\node [below left = -1mm and -53mm of pic, rotate = 0]{Engine speed $\left(\SI{}{rpm}\right)$};
			
			\end{tikzpicture}
			\par\end{centering}
		\caption{Engine operational space divisions and the corresponding linearisation points.}
		\label{fig:mod_grid}
	\end{figure}

	\begin{table}
		\centering
		\caption{States of the linear model.}\label{tab:states}
		\begin{tabularx}{0.77\columnwidth}{@{} cXl @{}}
			\hline
			State & Description & Units \\
			\hline
			$p_{\textrm{im}}$ & Intake manifold pressure perturbation & \si{kPa} \\
			$p_{\textrm{em}}$ & Exhaust manifold pressure perturbation & \si {kPa} \\
			$W_{\textrm{comp}}$ & Compressor mass flow rate perturbation & \si{g/s} \\
			$y_{\textrm{EGR}}$ & EGR rate perturbation & --\\
			\hline
		\end{tabularx}
	\end{table}

		\begin{table}
		\centering
		\caption{Inputs of the linear model.}\label{tab:inputs}
		\begin{tabularx}{0.77\columnwidth}{@{} cXl @{}}
			\hline
			State & Description & Units \\
			\hline
			$u_{\textrm{thr}}$ & Perturbation in throttle valve position& \si{\%} \\
			$u_{\textrm{EGR}}$ & Perturbation in EGR valve position & \si {\%} \\
			$u_{\textrm{VGT}}$ & Perturbation in VGT position & \si{\%} \\
			\hline
		\end{tabularx}
	\end{table}

	\begin{table}
	\centering
	\caption{Outputs of the linear model.}\label{tab:outputs}
	\begin{tabularx}{0.77\columnwidth}{@{} cXl @{}}
		\hline
		State & Description & Units \\
		\hline
		$p_{\textrm{im}}$ & Intake manifold pressure perturbation & \si{kPa} \\
		$y_{\textrm{EGR}}$ & EGR rate perturbation & --\\
		\hline
	\end{tabularx}
\end{table}
	
	The inputs for these models are the throttle position (percent closed), the EGR valve position (percent open) and the VGT vane position (percent closed). The steady state values are experimentally determined by sweeping through engine speeds and fuelling rates corresponding to the linearisation/model grid points, $\left(\omega_e^g,\dot{m}_f^g\right)$,  $\forall g\in \left\{\textrm{I},\,\textrm{II},\,\ldots,\,\textrm{XII} \right\}$. At each grid point, the steady state input values are obtained by allowing sufficient time for the transients to settle  and recording the actuator values applied by the ECU. The corresponding steady state and output values are also recorded at each operating point. These steady values denote the trimming conditions about which linear models need to be identified. The steady values for the input, state and output are denoted by the vectors ${\steady{u}}^g \in \mathbb{R}^3$, ${\steady{x}}^g \in \mathbb{R}^4$ and ${\steady{y}}^g \in \mathbb{R}^2$, respectively.
	
	System excitation is performed at each grid point, $\left(\omega_e^g,\dot{m}_f^g\right)$, using a pseudo-random input perturbation sequence of $\pm 15 \%$ of the steady state values, ${\steady{u}}^g$, applied simultaneously to all actuators listed in Table \ref{tab:inputs} to capture the state and output (see Table \ref{tab:outputs}) perturbations. The resulting input and output trajectories are divided into training and validation datasets and \textsc{Matlab}'s system identification toolbox \cite{Mathworks2017} is employed for model identification using the training dataset, while the validation dataset is used to determine whether or not the identification performance is adequate. The data for system identification is experimentally obtained from the test bench described in Section \ref{sec:RT_implementation}. 
		
	\begin{figure}
		\begin{centering}
			\begin{tikzpicture}
			\pgfkeys{/pgf/number format/.cd,fixed,precision=2}
			
			\node (pic) at (0,0)
			{\includegraphics[clip, trim = {0.00cm 0.00cm 0.0cm 0.0cm}]{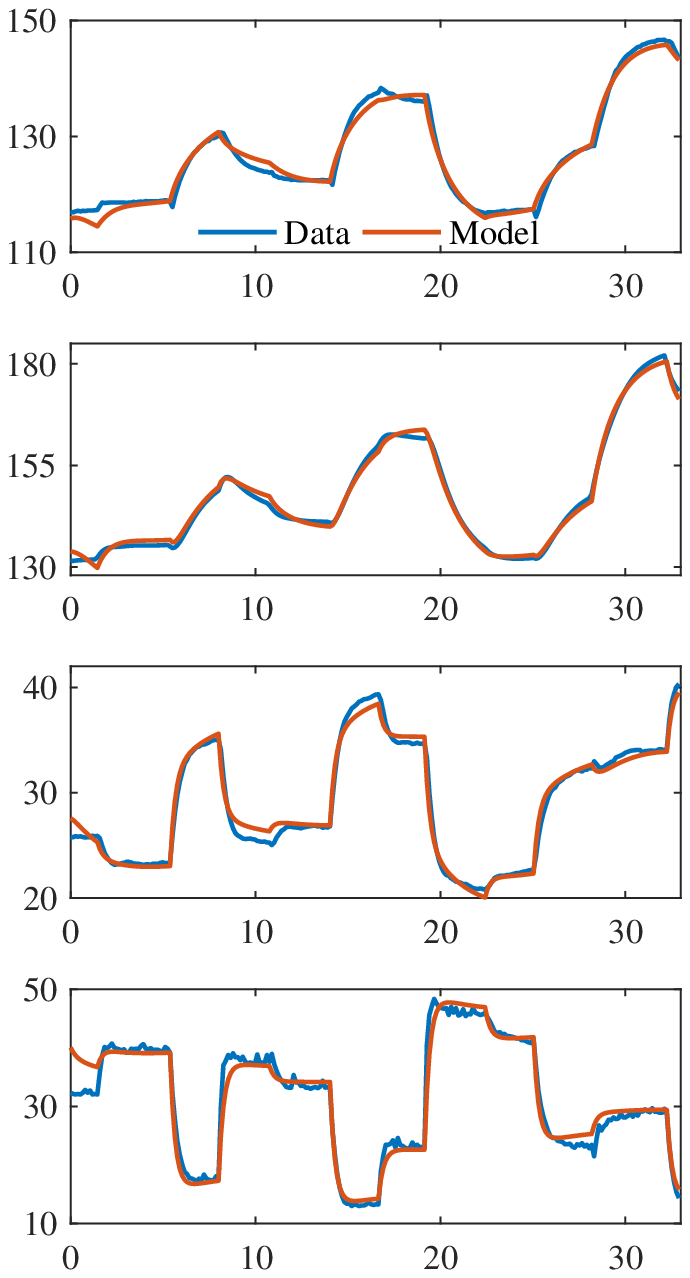}};

			\node [ rotate = 90, above left = -12mm and -4.5mm of pic] { {Normalised p$_{\textrm{im}}$ }};
			\node [ rotate = 90, above left = -45mm and -4.5mm of pic] { Normalised {p$_{\textrm{em}}$ }};
			\node [ rotate = 90, above left = -76mm and -4.5mm of pic] { Normalised {W$_{\textrm{comp}}$ }};
			\node [ rotate = 90, above left = -109mm and -4.5mm of pic] { Normalised {y$_{\textrm{EGR}}$}};
			\node [ rotate = 0, below   = -13mm  of pic] {Time (s)};
			
			
			\node (rect) [draw = none, fill = white, minimum width = 0.65cm, minimum height = 2.75cm, inner sep = 0pt, below  left = -142mm and -10.5mm of pic] {};
			
			\node [ below  left = -142.5mm and -11mm of pic] {\pgfmathparse{150/\pimbar}\pgfmathprintnumber\pgfmathresult};
			\node [ below  left = -130.5mm and -11mm of pic] {\pgfmathparse{130/\pimbar}\pgfmathprintnumber\pgfmathresult};
			\node [ below  left = -118.5mm and -11mm of pic] {\pgfmathparse{110/\pimbar}\pgfmathprintnumber\pgfmathresult};
			
			\node (rect) [draw = none, fill = white, minimum width = 0.65cm, minimum height = 2.75cm, inner sep = 0pt, below  left = -107mm and -10.5mm of pic] {};
			
			\node [ below  left = -107.5mm and -11mm of pic] {\pgfmathparse{180/\pimbar}\pgfmathprintnumber\pgfmathresult};
			\node [ below  left = -97mm and -11mm of pic] {\pgfmathparse{155/\pimbar}\pgfmathprintnumber\pgfmathresult};
			\node [ below  left = -87mm and -11mm of pic] {\pgfmathparse{130/\pimbar}\pgfmathprintnumber\pgfmathresult};
			
			\node (rect) [draw = none, fill = white, minimum width = 0.65cm, minimum height = 2.75cm, inner sep = 0pt, below  left = -75mm and -11mm of pic] {};
			
			\node [ below  left = -74.5mm and -11.5mm of pic] {\pgfmathparse{40/\wcbar}\pgfmathprintnumber\pgfmathresult};
			\node [ below  left = -64mm and -11.5mm of pic] {\pgfmathparse{30/\wcbar}\pgfmathprintnumber\pgfmathresult};
			\node [ below  left = -53mm and -11.5mm of pic] {\pgfmathparse{20/\wcbar}\pgfmathprintnumber\pgfmathresult};
			
			\node (rect) [draw =none, fill = white, minimum width = 0.65cm, minimum height = 2.75cm, inner sep = 0pt, below  left = -43mm and -11mm of pic] {};
			
			\node [ below  left = -44mm and -11.5mm of pic] {\pgfmathparse{50/\egrbar}\pgfmathprintnumber\pgfmathresult};
			\node [ below  left = -32mm and -11.5mm of pic] {\pgfmathparse{30/\egrbar}\pgfmathprintnumber\pgfmathresult};
			\node [ below  left = -20mm and -11.5mm of pic] {\pgfmathparse{10/\egrbar}\pgfmathprintnumber\pgfmathresult};
			
			
			\node (rect) [draw = none, fill = white, minimum width = 6.25cm, minimum height = 0.3cm, inner sep = 0pt, below  left = -114.55mm and -73mm of pic] {};
			
			\node [ rotate = 0, below  left = -115.5mm and -14mm of pic] {0};
			\node [ rotate = 0, below  left = -115.5mm and -33.5mm of pic] {10};
			\node [ rotate = 0, below  left = -115.5mm and -52mm of pic] {20};
			\node [ rotate = 0, below  left = -115.5mm and -71mm of pic] {30};

			\node (rect) [draw = none, fill = white, minimum width = 6.25cm, minimum height = 0.3cm, inner sep = 0pt, below  left = -81.8mm and -73mm of pic] {};
			
			\node [ rotate = 0, below  left = -82.75mm and -14mm of pic] {0};
			\node [ rotate = 0, below  left = -82.75mm and -33.5mm of pic] {10};
			\node [ rotate = 0, below  left = -82.75mm and -52mm of pic] {20};
			\node [ rotate = 0, below  left = -82.75mm and -71mm of pic] {30};

			\node (rect) [draw = none, fill = white,  minimum width = 6.25cm, minimum height = 0.3cm, inner sep = 0pt, below  left = -49.05mm and -73mm of pic] {};
			
			\node [ rotate = 0, below  left = -50mm and -14mm of pic] {0};
			\node [ rotate = 0, below  left = -50mm and -33.5mm of pic] {10};
			\node [ rotate = 0, below  left = -50mm and -52mm of pic] {20};
			\node [ rotate = 0, below  left = -50mm and -71mm of pic] {30};

			\node (rect) [draw = none, fill = white, minimum width = 6.25cm, minimum height = 0.3cm, inner sep = 0pt, below  left = -16.05mm and -73mm of pic] {};
			
			\node [ rotate = 0, below  left = -17mm and -14mm of pic] {0};
			\node [ rotate = 0, below  left = -17mm and -33.5mm of pic] {10};
			\node [ rotate = 0, below  left = -17mm and -52mm of pic] {20};
			\node [ rotate = 0, below  left = -17mm and -71mm of pic] {30};
			
			\end{tikzpicture}
			\vspace{-1cm}
			\par\end{centering}
		\protect\caption{Model fit for the training dataset at the grid point VI.}
		\label{fig:mod_fit_train}
	\end{figure}

	The system identification performance at the linearisation grid point VI, over the training and test datasets are shown in  Figs.~\ref{fig:mod_fit_train} and \ref{fig:mod_fit_test}, respectively. The model predictions are seen to follow the trend of the test data, however, as expected, discrepancies between the linear model predictions and the test data are observed. These disparities arise due to external disturbances, measurement errors and modelling errors as a consequence of using low-order discretised models. A pragmatic state disturbance set can now be estimated from the engine test bench data obtained for system identification. 
	
		\begin{figure}
		\begin{centering}
			\begin{tikzpicture}
			
			\pgfkeys{/pgf/number format/.cd,fixed,precision=2}
			
			\node (pic) at (0,0)
			{\includegraphics[clip, trim = {0.00cm 0.00cm 0.0cm 0.0cm}]{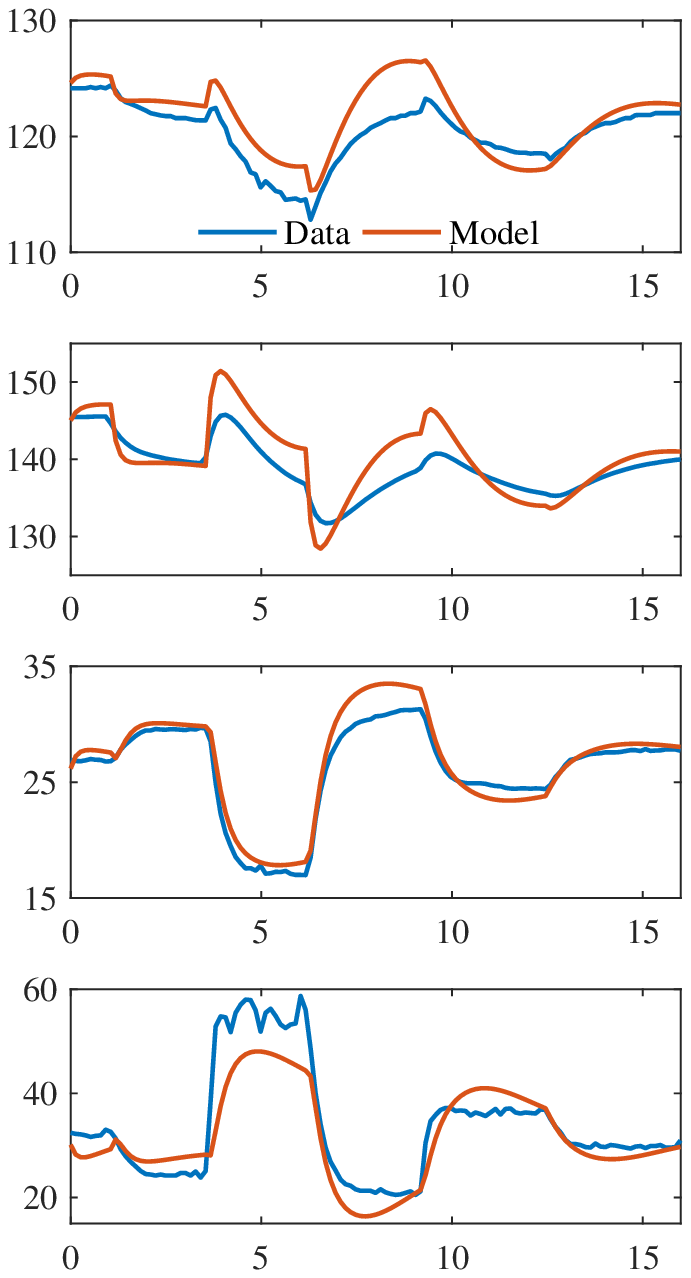}};
			
			\node [ rotate = 90, above left = -12mm and -4.5mm of pic] { {Normalised p$_{\textrm{im}}$ }};
			\node [ rotate = 90, above left = -45mm and -4.5mm of pic] { Normalised {p$_{\textrm{em}}$ }};
			\node [ rotate = 90, above left = -76mm and -4.5mm of pic] { Normalised {W$_{\textrm{comp}}$ }};
			\node [ rotate = 90, above left = -109mm and -4.5mm of pic] { Normalised {y$_{\textrm{EGR}}$}};
			\node [ rotate = 0, below   = -13mm  of pic] {Time (s)};

			
			\node (rect) [draw = none, fill = white, minimum width = 0.65cm, minimum height = 2.75cm, inner sep = 0pt, below  left = -142mm and -10.5mm of pic] {};
			
			\node [ below  left = -142.5mm and -11mm of pic] {\pgfmathparse{130/\pimbar}\pgfmathprintnumber\pgfmathresult};
			\node [ below  left = -130.5mm and -11mm of pic] {\pgfmathparse{120/\pimbar}\pgfmathprintnumber\pgfmathresult};
			\node [ below  left = -118.5mm and -11mm of pic] {\pgfmathparse{110/\pimbar}\pgfmathprintnumber\pgfmathresult};
			
			\node (rect) [draw = none, fill = white, minimum width = 0.65cm, minimum height = 2.75cm, inner sep = 0pt, below  left = -107mm and -10.5mm of pic] {};
			
			\node [ below  left = -105.5mm and -11mm of pic] {\pgfmathparse{150/\pimbar}\pgfmathprintnumber\pgfmathresult};
			\node [ below  left = -97.5mm and -11mm of pic] {\pgfmathparse{140/\pimbar}\pgfmathprintnumber\pgfmathresult};
			\node [ below  left = -90mm and -11mm of pic] {\pgfmathparse{130/\pimbar}\pgfmathprintnumber\pgfmathresult};
			
			\node (rect) [draw = none, fill = white, minimum width = 0.65cm, minimum height = 2.8cm, inner sep = 0pt, below  left = -76mm and -11mm of pic] {};
			
			\node [ below  left = -76.5mm and -11.5mm of pic] {\pgfmathparse{35/\wcbar}\pgfmathprintnumber\pgfmathresult};
			\node [ below  left = -65mm and -11.5mm of pic] {\pgfmathparse{25/\wcbar}\pgfmathprintnumber\pgfmathresult};
			\node [ below  left = -53mm and -11.5mm of pic] {\pgfmathparse{15/\wcbar}\pgfmathprintnumber\pgfmathresult};
			
			\node (rect) [draw =none, fill = white, minimum width = 0.65cm, minimum height = 2.75cm, inner sep = 0pt, below  left = -43mm and -11mm of pic] {};
			
			\node [ below  left = -44mm and -11.5mm of pic] {\pgfmathparse{60/\pimbar}\pgfmathprintnumber\pgfmathresult};
			\node [ below  left = -33.5mm and -11.5mm of pic] {\pgfmathparse{40/\pimbar}\pgfmathprintnumber\pgfmathresult};
			\node [ below  left = -22.5mm and -11.5mm of pic] {\pgfmathparse{20/\pimbar}\pgfmathprintnumber\pgfmathresult};

			
			\node (rect) [draw = none, fill = white, minimum width = 6.25cm, minimum height = 0.3cm, inner sep = 0pt, below  left = -114.55mm and -73mm of pic] {};
			
			\node [ rotate = 0, below  left = -115.5mm and -14mm of pic] {0};
			\node [ rotate = 0, below  left = -115.5mm and -33mm of pic] {5};
			\node [ rotate = 0, below  left = -115.5mm and -53mm of pic] {10};
			\node [ rotate = 0, below  left = -115.5mm and -72.5mm of pic] {15};

			\node (rect) [draw = none, fill = white, minimum width = 6.25cm, minimum height = 0.3cm, inner sep = 0pt, below  left = -81.8mm and -73mm of pic] {};
			
			\node [ rotate = 0, below  left = -82.75mm and -14mm of pic] {0};
			\node [ rotate = 0, below  left = -82.75mm and -33mm of pic] {5};
			\node [ rotate = 0, below  left = -82.75mm and -53mm of pic] {10};
			\node [ rotate = 0, below  left = -82.75mm and -72.5mm of pic] {15};

			\node (rect) [draw = none, fill = white,  minimum width = 6.25cm, minimum height = 0.3cm, inner sep = 0pt, below  left = -49.05mm and -73mm of pic] {};
			
			\node [ rotate = 0, below  left = -50mm and -14mm of pic] {0};
			\node [ rotate = 0, below  left = -50mm and -33mm of pic] {5};
			\node [ rotate = 0, below  left = -50mm and -53mm of pic] {10};
			\node [ rotate = 0, below  left = -50mm and -72.5mm of pic] {15};

			\node (rect) [draw = none, fill = white, minimum width = 6.25cm, minimum height = 0.3cm, inner sep = 0pt, below  left = -16.05mm and -73mm of pic] {};
			
			\node [ rotate = 0, below  left = -17mm and -14mm of pic] {0};
			\node [ rotate = 0, below  left = -17mm and -33mm of pic] {5};
			\node [ rotate = 0, below  left = -17mm and -53mm of pic] {10};
			\node [ rotate = 0, below  left = -17mm and -72.5mm of pic] {15};
			
			\end{tikzpicture}
			\par\end{centering}
		\vspace{-1cm}
		\protect\caption{Model fit for the test dataset at the grid point VI.}
		\label{fig:mod_fit_test}
	\end{figure}

	
	In this work, the disturbance set at each grid point, $\mathcal{W}^g$, is chosen as the hypercube (primarily for its simplicity) that contains the discrepancies between the linear model predictions and the actual engine data from the test dataset. The disturbance sets are convex, compact (closed and bounded) and include the origin. Fig.~\ref{fig:W_set} shows the error distributions between the model predictions and the measurements in each state channel at the operating point VI.

	\begin{figure}
		\begin{centering}
			\begin{tikzpicture}
			\node (pic) at (0,0)
			{\includegraphics[clip, trim = {0.00cm 0.00cm 0.0cm 0.0cm}]{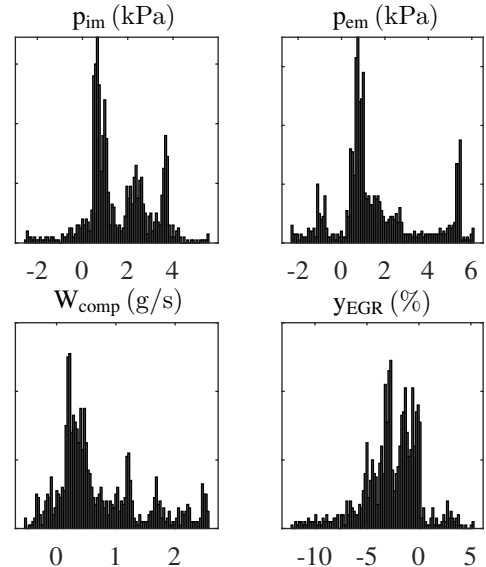}};
			
			\node [ rotate = 0, above left = -2mm and -24mm of pic] { {p$_{\textrm{im}}\,(\SI{}{kPa})$ }};
			\node [ rotate = 0, above left = -2mm and -59mm of pic] { {p$_{\textrm{em}}\,(\SI{}{kPa})$ }};
			\node [ rotate = 0, above left = -40mm and -25mm of pic] { {W$_{\textrm{comp}}\,(\SI{}{g/s})$ }};
			\node [ rotate = 0, above left = -40mm and -58mm of pic] { {y$_{\textrm{EGR}}\,(\%)$}};
			
			\end{tikzpicture}
			\par\end{centering}
		\protect\caption{Distribution of state disturbances at the operating point VI.}
		\label{fig:W_set}
	\end{figure}
	
	The discrepancies between model predictions and the actual engine behaviour can be represented as additive uncertainties on the states. The resulting linear representation of the diesel airpath trimmed about a given grid point, $\left(\omega_e^g,\dot{m}_f^g\right)$, is
	\begin{subequations}
		\begin{eqnarray}
		x_{k+1} & = & A_{}^gx_{k}+B^g_{}u_{k}+w_{k},\\
		y_{k} & = & C_{}^gx_{k} + D_{}^gu_{k},
		\end{eqnarray}\label{eq:lsys}
	\end{subequations}
	
	\noindent where the perturbed states, $x_{k} \coloneqq \left\{p_{\textrm{im}},\: p_{\textrm{em}},\: W_{\textrm{comp}},\: y_{\textrm{EGR}}\right\}$ are respectively the perturbations in the intake and the exhaust manifold pressures, the flow rate through compressor and the EGR rate about ${\steady{x}}^g$; $y_{k}\coloneqq\left\{p_{\textrm{im}},\: y_{\textrm{EGR}}\right\}$ are the output perturbations; the perturbed control inputs, $u_{k}\coloneqq \left\{u_{\textrm{thr}},\: u_{\textrm{EGR}},\: u_{\textrm{VGT}} \right\}$ are respectively the perturbations in the throttle and the EGR valve positions and the VGT position about ${\steady{u}}^g$; and $w_k\in\mathbb{R}^{4}$ is an unknown but bounded state disturbance, contained in the disturbance set ${\mathcal{W}^g} \coloneqq  \left\{ w|\zeta^g w\leq\theta^g,\,\zeta^g\in\mathbb{R}^{a\times 4},\,\theta^g\in\mathbb{R}^{a}\right\}$, with $a \in 2\mathbb{Z}^{+}$.

	\begin{assm} Each pair $\left(A^g,\,B^g\right)$ is stabilisable. \label{ass:stab}
	\end{assm}

	\section{Model Predictive Controller Development}
	\label{sec:mpc}
	
	The objective of the diesel airpath controller is to regulate the outputs, the intake manifold pressure and the EGR rate to their reference values, whilst adhering to an upper bound constraints on the intake and exhaust manifolds pressures for safety and reliability requirements, physical limitations of the actuators and slew rate constraints on the actuators in the presence of state disturbances. In this section, the online MPC and offline CT optimisation problem that are required to achieve the desired controller objectives are described.

	\subsection{Online MPC Optimisation}
	Let the polytopic pointwise-in-time state and input constraints be given as
	\begin{subequations}
		\begin{align}
		&x(k)\in\mathcal{X}\coloneqq \left\{ x|Ex\leq f,\, E\in\mathbb{R}^{q\times 4},\, f\in\mathbb{R}^{q}\right\}, \\
		&u(k)\in\mathcal{U}\coloneqq\left\{ u|Gu\leq h,\, G\in\mathbb{R}^{r\times 3},\, h\in\mathbb{R}^{r}\right\},
		\end{align}\label{eq:con}
	\end{subequations}
	
	\noindent where $x(k) = x_k + \steady{x}^g$;  $u(k) = u_k + \steady{u}^g$; $q$ and $r$ represent the number of facets of $\mathcal{X}$ and $\mathcal{U}$, respectively.
	
	\begin{assm} There exists a positively invariant set, $\mathcal{X}_{f} \coloneqq  \left\{ x|S x\leq t,\,S\in\mathbb{R}^{p\times 4},\,t\in\mathbb{R}^{p}\right\} \subseteq \mathcal{X}$, under a stabilising controller $\kappa_f\left(x\right)$ $\forall \left(A^g,\, B^g\right) $, $g\in \left\{\textrm{I},\,\textrm{II},\,\ldots,\,\textrm{XII} \right\}$ \cite{Broomhead2014}; where the number of facets of $\mathcal{X}_{f}$ is denoted by $p$.\label{ass:kappa}
	\end{assm}

	The novel MPC formulation proposed in \cite{Shekhar2017} and \cite{Sankar2017} with exponentially decaying envelopes to aid fast calibration is summarised below. The MPC cost function is defined as
	\begin{align}
	V_{N}\left(x(k),\,\boldsymbol{\pv{u}_k}\right) \coloneqq & \sum_{j=0}^{N-1}\left\Vert W^g\, Y_{k+j|k}\right\Vert^{2}     \nonumber\\
	& +	\sum_{j=0}^{N-2} \left\Vert\boldsymbol{\epsilon^{g}} W^g\left(\pv{y}_{k+1+j|k}-\Gamma^g \pv{y}_{k+j|k}\right)\right\Vert^{2} \nonumber \\
	& +	\sum_{j=0}^{N-1} \gamma \left\Vert R \pv{u}_{k+j|k}\right\Vert^{2},
	\label{eq:mpccost}
	\end{align}
	
	\noindent where $N$ is the prediction horizon;  $\boldsymbol{\pv{u}_k} \coloneqq \left\{\pv{u}_{k|k},\pv{u}_{k+1|k},\ldots,\pv{u}_{k+N-1|k}\right\}$ is the input sequence applied over the horizon; the height of the envelope at a prediction step $j$ is represented by $Y_{k+j|k} \in \mathbb{R}^2$. Given an initial envelope height at the current time step, $Y_{k|k}$, and a sampling time, $T_s$, the decay envelope is enforced by
	\begin{equation}
	Y_{k+1+j|k}=\Gamma^g Y_{k+j|k},\,\forall j\in\mathbb{Z}_{\left[0:N-2\right]},
	\end{equation}
	\noindent where $\Gamma^g \coloneqq \exp \left(-\diag \{T_s/\tau_{\boost}^g, T_s/\tau_{\EGR}^g \} \right)$. The envelope time constants corresponding to the output channels for a given $\left(\omega_e^g,\dot{m}_f^g\right)$, $\tau_{\textrm{boost}}^g$ and $\tau_{\textrm{EGR}}^g$, are the primary tuning parameters used to shape the output response. 
	
	The first term in the cost function \eqref{eq:mpccost} penalises the envelope heights for output prioritisation through the matrix,
	\begin{equation}
	W^{g} \coloneqq \diag\left\{w^{g}/w_{\textrm{boost}},\,\left(1-w^{g}\right)/w_{\textrm{EGR}}\right\},
	\end{equation}
	for the envelope priority parameter $w^{g} \in [0,\,1]$ and normalisation constants $w_{\textrm{boost}}$ and $w_{\textrm{EGR}}$, that are chosen to scale the output variation to between zero and one over the engine operational range. The parameter $w^{g}$ is used to fine-tune the relative time constants between the two output channels i.~e.,	$w^{g}$ defines the relative priority in minimising the envelope height corresponding to one output at the expense of increase in the envelope height of the other output channel. A value of $w^{g} > 0.5$ will prioritise intake manifold pressure over EGR rate.
	
	The secondary objective is to encourage smooth exponential decay of the outputs through the second term in the MPC cost function \eqref{eq:mpccost}. The smoothness term penalises output deviation from an ideal exponential decay towards zero (or equivalently, the actual outputs decaying towards their steady state values). The parameter $\boldsymbol{\epsilon^{g}} \coloneqq \diag\{\epsilon_{\boost}^{g}, \epsilon_{\EGR}^{g}$\}, is a diagonal matrix that is used to adjust the smoothness on each output channel. 
	
	Finally, the third term in the cost function \eqref{eq:mpccost}, is used for input regularisation, through the constant $\gamma \in [0,\,\infty)$ and the matrix $R$. This term encourages the input perturbations to tend towards zero, which is equivalent to the actual inputs approaching their desired steady state values for the given operating point. In this work, $\gamma$ is chosen sufficiently small such that the input regularisation is the tertiary objective and does not dominate the envelope and smoothness cost terms.
	
	To handle transient operation of the engine over a drive cycle, a switched LTI-MPC strategy is proposed where the controller at the nearest grid point (see Fig.~\ref{fig:mod_grid}), $\left(\omega_e^g,\dot{m}_f^g\right)$, to the current operating point, $\left(\omega_e,\dot{m}_f\right)$, is chosen at each time instant. Therefore, the system matrices identified at the grid point $\left(\omega_e^g,\dot{m}_f^g\right)$, $A^{g}$, $B^{g}$, $C^{g}$ and $D^{g}$; the corresponding tuning parameters, $\tau_{\boost}^{g}$, $\tau_{\EGR}^{g}$, $w^{g}$, $\epsilon_{\boost}^{g}$ and $\epsilon_{\EGR}^{g}$; and the constraint tightening margins reserved over the horizon for the state and input constraints,  $\boldsymbol{\sigma}^{g*}\coloneqq \left\{ \sigma_{0}^{g*},\,\sigma_{1}^{g*},\,\ldots,\,\sigma_{N}^{g*}\right\}$ and $\boldsymbol{\mu}^{g*}\coloneqq \left\{ \mu_{0}^{g*},\,\mu_{1}^{g*},\,\ldots,\,\mu_{N-1}^{g*}\right\}$, respectively, used within the MPC optimisation problem, $\mathcal{P}_N\left(x(k),\,g\right)$, are updated at each sampling instant. The steady state values (or trimming conditions), ${\steady{x}_k}$, ${\steady{u}_k}$ and ${\steady{y}_k}$, for the current operating point, $\left(\omega_e,\dot{m}_f\right)$, are determined using linear interpolation of the identified steady state values at the neighbouring grid points, ${\steady{x}^g}$, ${\steady{u}^g}$ and ${\steady{y}^g}$. The regions around the grid points as shown in Fig.~\ref{fig:mod_grid} indicate the operating points for which each `local' controller will be active.
	
	The control law, $\kappa_N\left(x(k)\right) = \pv{u}^*_{k|k} + {\steady{u}_k}$, is used to determine the control input applied to the engine at each sampling instant, where $\pv{u}^*_{k|k}$ is the first element of the optimal input sequence, $\boldsymbol{\pv{u}_k}^* \coloneqq\left\{\pv{u}_{k|k}^*,\pv{u}_{k+1|k}^*,\ldots,\pv{u}_{k+N-1|k}^*\right\}$, obtained by solving:
	\begin{subequations}
		\begin{align}
		\mathcal{P}_N\left(x(k),\,g\right): & \underset{\boldsymbol{\pv{u}_k},\,Y_{k|k}}{\min}\,\,\,  V_{N}\left(x(k),\,\boldsymbol{\pv{u}_k}\right) \\
		\textrm{s.t } & \forall j\in\mathbb{Z}_{\left[0:N-1\right]} \nonumber\\
		& \pv{x}_{k|k}=x(k) - {\steady{x}_k} \label{eq:initst}\\
		& \pv{x}_{k+1+j|k}=A^g\pv{x}_{k+j|k}+B^g\pv{u}_{k+j|k} \label{eq:dyn1}\\
		& \pv{y}_{k+j|k}=C^g\pv{x}_{k+j|k}+D^g\pv{u}_{k+j|k} \label{eq:dyn2}\\
		& Y_{k+1+j|k}=\Gamma^g Y_{k+j|k} ,\,\forall j\in\mathbb{Z}_{\left[0:N-2\right]} \label{eq:decay}\\
		& Y_{k|k-1}\leq Y_{k|k} \label{eq:decaylimit}\\
		& -Y_{k+j|k}\leq \pv{y}_{k+j|k}\leq Y_{k+j|k} \label{eq:decaycon}\\
		& E\left(\pv{x}_{k+j|k}+\steady{x}_k\right)\leq f-\sigma^{g*}_{j},\,\forall j\in\mathbb{Z}_{\left[1:N-1\right]}\label{eq:statecon}\\
		& G\left(\pv{u}_{k+j|k}+\steady{u}_k\right)\leq h- \mu^{g*}_{j}\label{eq:inpcon} \\
		& S\left(\pv{x}_{k+N|k}+\steady{x}_k\right)\leq t- \sigma^{g*}_{N}\label{eq:termcon}\\
		& \left|\pv{u}_{k|k-1}^* - \pv{u}_{k|k}\right| \leq \delta \boldsymbol{1}_{3} \label{eq:inpslew1}\\
		& \Delta \pv{u}_{k+j|k} \leq \delta \boldsymbol{1}_3,\,\forall j\in\mathbb{Z}_{\left[ 0:N-2\right]}\label{eq:inpslew2},
		\end{align}
		\label{eq:mpc}
	\end{subequations}
	
	\noindent and $\Delta \pv{u}_{k+j|k} \coloneqq \left|\pv{u}_{k+1+j|k} - \pv{u}_{k+j|k}\right|$; the scalar, $\delta$, defines the maximal allowable change in the actuator position in one sampling time.

	The initial condition and nominal system dynamics are included in \eqref{eq:initst}-\eqref{eq:dyn2}. The constraint \eqref{eq:decaylimit} enforces the initial height of the envelope between successive steps of the MPC iteration to decay at a maximum rate $\Gamma^g$, to prevent one of the outputs decaying too abruptly, thereby, potentially resulting in a poor response on the other output channel. The output perturbations over the horizon are restricted to remain within exponentially decaying envelopes by \eqref{eq:decaycon}. Finally, terminal state constraint and slew rate constraints on the inputs are imposed through \eqref{eq:termcon}-\eqref{eq:inpslew2}. 
	
	For a given model grid point, $\left(\omega_e^g,\dot{m}_f^g\right)$, to handle the uncertainties, consider there exists tightened constraints, $\left\{\mathcal{X}_0^g,\,\mathcal{X}_1^g,\,\ldots,\,\mathcal{X}_{N}^g\right\}$ and $\left\{\mathcal{U}_0^g,\,\mathcal{U}_1^g,\,\ldots,\,\mathcal{U}_{N-1}^g\right\}$, for the state and input constraints, respectively. The tightened constraints are obtained by reserving margins $\boldsymbol{\sigma}^{g*}$ and $\boldsymbol{\mu}^{g*}$, from the facets of the polytopic state and input constraints, respectively, as in equations \eqref{eq:statecon} and \eqref{eq:inpcon}.

	Algorithm~\ref{alg:drivecycle} defines the control strategy applied to the engine air path at each time instant during a transient operation over a drive cycle.
	\begin{alg}[MPC over drive cycle] \label{alg:drivecycle}
		At time step $k$, given the current operating condition, $\left(\omega_e,\dot{m}_f\right)$, and state, $x(k)$:
		\begin{enumerate}
			\item Determine the nearest grid point $\left(\omega_e^g,\dot{m}_f^g\right)$ and hence, define the corresponding model parameters -  $A^{g}$, $B^{g}$, $C^{g}$ and $D^{g}$; and the tuning parameters - $\tau_{\boost}^{g}$, $\tau_{\EGR}^{g}$, $w^{g}$, $\epsilon_{\boost}^{g}$ and $\epsilon_{\EGR}^{g}$; and the tightening margins - $\boldsymbol{\sigma}^{g*}$ and $\boldsymbol{\mu}^{g*}$.
			\item Determine the steady state values, ${\steady{x}_k}$, ${\steady{u}_k}$ and ${\steady{y}_k}$, by interpolating linearly from the neighbouring grid points.
			\item Evaluate $\boldsymbol{\pv{u}}^*_k$ as the solution to $\mathcal{P}_N(x(k),\,g)$.
			\item Apply the input $u(k) = \pv{u}^*_{k|k} + {\steady{u}_k}$ to the actuators.
			\item Set $k \leftarrow k + 1$ and return to step 1.
		\end{enumerate}
	\end{alg}
	
	\begin{rem}
		The maximum input change magnitude, $\delta$, in \eqref{eq:inpslew1}-\eqref{eq:inpslew2} is chosen such that $\delta$ is at least the maximum difference in the reserved input tightening margin between successive time instants along the horizon, i.e., $\delta \geq \max \left( \|\mu_{j+1}^{g*} - \mu_{j}^{g*}\|_1\right),\,\forall j \in\mathbb{Z}_{\left[0:N-2\right]}$ and $g\in \left\{\textrm{I},\,\textrm{II},\,\ldots,\,\textrm{XII} \right\}$. \label{rem:delta}
	\end{rem}
	
	\begin{rem} \label{rem:benefit}
		Incorporation of the output envelopes to the MPC formulation and the appropriate cost function parameterisation significantly reduces the number of tuning parameters when compared to conventional MPC, thereby, greatly reducing the calibration effort. For the nominal identified system (without the disturbance term $w_k$ in \eqref{eq:lsys}), used in the conventional QP MPC framework, a standard quadratic cost on the states and inputs with semidefinite matrices and no terminal cost would require 16 parameters per model grid point in the symmetric $Q$ and $R$ matrices to be calibrated, whereas, the proposed calibration-friendly MPC formulation requires at most five parameters namely, $\tau_{\boost}^{g}$, $\tau_{\EGR}^{g}$, $w^{g}$, $\epsilon_{\boost}^{g}$ and $\epsilon_{\EGR}^{g}$, per grid point to be tuned. Furthermore, there is an intuitive relationship between these parameters and the open-loop output transient response.
	\end{rem}
	
	\begin{rem} 
		While in transient operation over a drive cycle, at each sampling instant, the proposed MPC with exponential envelopes solves a regulation problem. The calibration of the controllers is dependent on the closed-loop response obtained for the set point problem. This provides a structured way of approaching drive cycle calibration by isolating an `under-performing' controller and systematically target its calibration.
	\end{rem}
	
	The smoothness term in the proposed MPC cost function \eqref{eq:mpccost} involves a weighted difference of the output predictions at consecutive time steps over the horizon. Hence, monotonicity of the cost function required for establishing asymptotic stability \cite{Mayne2000} is achieved only for certain choices of the design parameters (tuning parameters), thereby, establishing practical stability of the closed-loop system about the origin. 
	
	\begin{assm} The terminal controller satisfies the tightened input constraint at the terminal step of the horizon, i.e., $\kappa_f({x})  \in \mathcal{U}_{N-1}^{g},\, \forall {x}\in\mathcal{X}_f$ and $g\in \left\{\textrm{I},\,\textrm{II},\,\ldots,\,\textrm{XII} \right\}$. \label{ass:kappa_ct}
	\end{assm}

	\begin{lemma}[Local asymptotic stability \cite{Khalil2002}]
		A steady state $\bar{x}_{0}$ is  locally asymptotically stable if there exists a class $\mathcal{K}\mathcal{L}$ function $\beta(\cdot,\cdot)$ such that 
		\begin{equation}
			\left|x\left(k\right)-\bar{x}_{0}\right| \leq \beta \left(\left|x\left(0\right)-\bar{x}_{0}\right|,k\right),\forall \left(x\left(0\right)-\bar{x}_{0}\right)<c, k \geq 0.
		\end{equation}
		\label{lem:lemma1}
	\end{lemma}
	
	\begin{theorem}[Practical stability]
		Consider the system represented by \eqref{eq:lsys}, subjected to the constraints \eqref{eq:con} and let the Assumptions \ref{ass:stab}, \ref{ass:kappa} and \ref{ass:kappa_ct} hold. Let  $\mathbb{X}_N$ be the feasible region for $\mathcal{P}_N(x(k),\,g)$. Then given a constant trim point $\bar{x}_0$, $x(0)\in\mathbb{X}_N$ and the control law $\kappa_N\left(x(k)\right)$, there exists a class $\mathcal{K}\mathcal{L}$ function $\beta(\cdot,\cdot)$ such that $\forall k\geq0$:
		\begin{equation}		
		\left| x(k) - \bar{x}_{0}\right| \leq \beta\left(\left|x(0) - \bar{x}_{0}\right|,\, k\right) + \mathcal{O} \left(\left\Vert\boldsymbol{\epsilon^{g}}\right\Vert^{2} + \left\Vert \kappa_f \right\Vert^{2} \right).
		\label{eq:stab}
		\end{equation}
		\label{thm:stab}
	\end{theorem}	

	\begin{proof}
		At a given time instant $k$, the optimal control sequence, $\boldsymbol{\pv{u}}_k^*=\left\{\pv{u}_{k|k}^*,\,\ldots,\,\pv{u}_{k+N-1|k}^*\right\}$, obtained by solving the MPC optimisation problem, $\mathcal{P}_N\left(x(k),\,g\right)$, in \eqref{eq:mpc}, satisfies the tightened input constraint i.e., $\pv{u}_{k+j|k}^* \in \mathcal{U}^g_j$ for $j\in\mathbb{Z}_{\left[0:N-1\right]}$ and the corresponding optimal state sequence, $\boldsymbol{\pv{x}}_k^*=\left\{\pv{x}_{k|k}^*,\,\ldots,\,\pv{x}_{k+N|k}^*\right\}$, satisfies the tightened state and terminal constraints, i.e., $\pv{x}_{k+j|k}^* \in \mathcal{X}_j^g$ for $j\in\mathbb{Z}_{\left[0:N-1\right]}$ and $\pv{x}_{k+N|k}^* \in \mathcal{X}_f$. Let the corresponding optimal output and envelope sequences be $\boldsymbol{\pv{y}}_k^*=\left\{\pv{y}_{k|k}^*,\,\ldots,\,\pv{y}_{k+N-1|k}^*\right\}$ and $\boldsymbol{{Y}}_k^*=\left\{{Y}_{k|k}^*,\,\ldots,\,{Y}_{k+N-1|k}^*\right\}$. Consider the following candidate control, output and envelope sequences, respectively, $\boldsymbol{\pv{u}}_{k+1}^0=\left\{\pv{u}_{k+1|k+1}^0,\,\ldots,\,\pv{u}_{k+N|k+1}^0\right\}$,  $\boldsymbol{\pv{y}}_{k+1}^0=\left\{\pv{y}_{k+1|k+1}^0,\,\ldots,\,\pv{y}_{k+N|k+1}^0\right\}$ and $\boldsymbol{{Y}}_{k+1}^0=\left\{{Y}_{k+1|k+1}^0,\,\ldots,\,{Y}_{k+N|k+1}^0\right\},$ for the problem \eqref{eq:mpc} at $k+1$, where $\forall j \in \mathbb{Z}_{\left[0:N-2\right]}$,
		\begin{subequations}
			\begin{align}
			& \pv{u}_{k+1+j|k+1}^0 = \pv{u}_{k+1+j|k}^*, \\
			& \pv{y}_{k+1+j|k+1}^0 = \pv{y}_{k+1+j|k}^*, \\
			& {Y}_{k+1+j|k+1}^0 = {Y}_{k+1+j|k}^*,\\
			& \pv{u}_{k+N|k+1}^0 = \kappa_f\left(\pv{x}_{k+N|k}^*\right),\\
			& \pv{y}_{k+N|k+1}^0 = C^g \pv{x}_{k+N|k}^* + D^g \pv{u}_{k+N|k+1}^0,\\
			& {Y}_{k+N|k+1}^0 = \Gamma^g {Y}_{k+N-1|k}^*.
			\end{align}
		\end{subequations}
		
		Let the corresponding candidate state sequence be $\boldsymbol{\pv{x}}_{k+1}^0=\left\{\pv{x}_{k+1|k+1}^0,\,\ldots,\,\pv{x}_{k+1+N|k+1}^0\right\}$, with  
		\begin{subequations}
			\begin{align}
			& \pv{x}_{k+1+j|k+1}^0 = \pv{x}_{k+1+j|k}^*,\,  \forall j \in \mathbb{Z}_{\left[0:N-1\right]},\\
			& \pv{x}_{k+1+N|k+1}^0 = A^g \pv{x}_{k+N|k}^* + B^g \pv{u}_{k+N|k+1}^0.
			\end{align}
		\end{subequations}

		The candidate control sequence is feasible for $\mathcal{P}_N\left(x(k+1),\,g\right)$ as:
		\begin{itemize}
			\item the terminal set $\mathcal{X}_f$ is positively invariant under the terminal controller, $\kappa_f(\pv{x})$, when $\pv{x} \in \mathcal{X}_f$, where $\kappa_f(\pv{x})$  satisfies the tightened input constraint, $\mathcal{U}_{N-1}^{g}$, and 
			\item a $N_{np} \leq N$ step nilpotent disturbance feedback policy is used.
		\end{itemize}

		The cost function of the problem $\mathcal{P}_N\left(x(k+1),\,g\right)$ for the feasible candidate control sequence is	
			\begin{align}
			& V_{N}\left(x(k+1),\, \boldsymbol{\pv{u}}_{k+1}^0\right) = V_{N}\left(x(k),\, \boldsymbol{\pv{u}}_k^*\right)   \nonumber\\
			&  + \left\Vert W^{g}\, \Gamma^{g} Y_{k+N-1|k}^*\right\Vert^{2} - \left\Vert W^{g}\, Y_{k|k}^*\right\Vert^{2}    \nonumber\\
			&  +	\left\Vert\boldsymbol{\epsilon^{g}} W^{g}\left(\pv{y}_{k+N|k+1}^0-\Gamma^{g} \pv{y}_{k+N-1|k}^*\right)\right\Vert^{2}    \nonumber\\
			&-	\left\Vert\boldsymbol{\epsilon^{g}} W^{g}\left(\pv{y}_{k+1|k}^*-\Gamma^{g} \pv{y}_{k|k}^*\right)\right\Vert^{2} \nonumber \\ 
			&   + \gamma \left\Vert R\kappa_f\left(\pv{x}_{k+N|k}^*\right) \right\Vert^{2}  - \gamma \left\Vert R\pv{u}^*_{k|k} \right\Vert^{2}.	
			\end{align}

		The optimal cost at $k+1$, therefore, for the optimal control sequence, $\boldsymbol{\pv{u}}_{k+1}^*$,  is	
		\begin{align}
		&V_{N}\left(x(k+1),\, \boldsymbol{\pv{u}}_{k+1}^*\right) \leq V_{N}\left(x(k+1),\, \boldsymbol{\pv{u}}_{k+1}^0\right) \nonumber\\
		&\leq V_{N}\left(x(k), \,\boldsymbol{\pv{u}}_k^*\right)    + \mathcal{O} \left(\left\Vert\boldsymbol{\epsilon^{g}}\right\Vert^{2} + \left\Vert \kappa_f \right\Vert^{2} \right).
		\label{eq:thm_prf}
		\end{align}

Eq.~\eqref{eq:stab} follows from \eqref{eq:thm_prf} by considering the definite positiveness of the optimal cost function $V_{N}\left(x\left(k\right),\,\boldsymbol{\pv{u}}_k^*\right)$ and its non- increasing evolution. Then practical stability of the trim point $\bar{x}_0$ is established by virtue of Lemma \ref{lem:lemma1} and \eqref{eq:stab}.

	\end{proof}

	The CT policy required for handling the uncertainties due to modelling errors and controller switching under the LTI-MPC architecture is proposed after stating the underlying assumption about the sizes of the state and input feasible sets.	
	\begin{assm} The sets $\mathcal{X}$ and $\mathcal{U}$ are sufficiently large such that non-empty tightened state and input constraints are obtained at the terminal step of the horizon while performing constraint tightening for a $N$-step predictive controller.\label{ass:non-empty}
	\end{assm}
	
	For a given controller at a model grid point, to provide recursive feasibility guarantees in the presence of system disturbances and uncertainties due to controller switching, the tightened state and input constraints under the constraint tightening approach with a disturbance feedback parameterisation are defined $\forall j\in\mathbb{Z}_{\left[0:N-2\right]}$, as	
	\begin{subequations}
		\begin{align}
		&{\mathcal{\ltvpv{X}}}_{0}^{g}\coloneqq \mathcal{X}, \,{\mathcal{\ltvpv{X}}}_{j+1}^{g}\coloneqq {\mathcal{\ltvpv{X}}}_{j}^{g}\ominus L_{j}^{g}\mathcal{W}^{g} \ominus \mathcal{N}_j \ominus \mathcal{S}_j \ominus \Delta\mathcal{X},\label{eq:ltv_st}\\
		&{\mathcal{\ltvpv{X}}}_{N}^{g}\coloneqq {\mathcal{{X}}}_{f}\ominus L_{N-1}^{g}\mathcal{W}^{g} \ominus \mathcal{N}_{N-1} \ominus \mathcal{S}_{N-1} \ominus \Delta\mathcal{X},\label{eq:ltv_termst}\\
		&{\mathcal{\ltvpv{U}}}_{0}^{g}\coloneqq \mathcal{U}, \, {\mathcal{\ltvpv{U}}}_{j+1}^{g}\coloneqq {\mathcal{\ltvpv{U}}}_{j}^{g}\ominus P_{j}^{g}\mathcal{W}^{g} \ominus \mathcal{M}_j \ominus \Delta\mathcal{U}, \label{eq:ltv_ip} 
		\end{align}\label{eq:ct_ltv}
	\end{subequations}
	
	\noindent where at a prediction step $j$, $P_{j}^g$  is the CT policy matrix providing direct feedback on the disturbance \cite{Kuwata2007}; $L_{j}^g$ is the disturbance transition matrix with $L_{0}^g\coloneqq I_{n}$ and
	\begin{align}
	L_{j+1}^g\coloneqq A^gL_{j}^g+B^gP_{j}^g,\,\forall j\in\mathbb{Z}_{\left[0:N-2\right]};
	\end{align}
	
	\noindent the sets $\mathcal{N}_j$ and $\mathcal{S}_j$ contain the uncertainties in state predictions and disturbance propagation, respectively, due to controller switching; $\mathcal{M}_j$ includes the feedback actions to reject those uncertainties in state predictions because of change in model and the sets $\Delta\mathcal{X}$ and $\Delta\mathcal{U}$ contain all possible deviations of steady state and input values, respectively.

	\begin{theorem}[Robust feasibility under switched LTI approach] Consider that Assumptions \ref{ass:stab}-\ref{ass:non-empty} hold. Let the tightening margins required to satisfy the constraint tightening policy in \eqref{eq:ct_ltv} for the state and input constraints, respectively, be $\boldsymbol{\ltvpv{\sigma}}^{g} \coloneqq \left\{ \ltvpv{\sigma}_{0}^{g},\,\ldots,\,\ltvpv{\sigma}_{N}^{g}\right\} $ and $\boldsymbol{\ltvpv{\mu}}^{g} \coloneqq \left\{ \ltvpv{\mu}_{0}^{g},\,\ldots,\,\ltvpv{\mu}_{N-1}^{g}\right\}$ and substitute $\boldsymbol{\sigma}^{g*} = \boldsymbol{\ltvpv{\sigma}}^{g}$ and $\boldsymbol{\mu}^{g*} = \boldsymbol{\ltvpv{\mu}}^{g}$ in \eqref{eq:mpc}. If $\mathcal{P}_{N}\left(x(k),\,g\right)$ has a feasible solution, then subsequent optimisation problems $\mathcal{P}_{N}  \left(x(k+j),\,g'\right)$, are feasible $\forall \,j>0$,  where $g'\in \left\{\textrm{I},\,\textrm{II},\,\ldots,\,\textrm{XII} \right\}$ represents a model grid point.	
	\end{theorem}
	
	\begin{proof} Robust feasibility proof of the MPC optimisation problem \eqref{eq:mpc} under the switched LTI architecture is based on recursion, showing that feasibility of  $\mathcal{P}_{N}  \left(x(k),\,g\right)$ implies feasibility of $\mathcal{P}_{N}  \left(x(k+1),\,g'\right)$. Feasibility of $\mathcal{P}_{N}  \left(x(k+1),\,g'\right)$ is proven by showing the feasibility of a candidate solution constructed from the solution of $\mathcal{P}_{N}  \left(x(k),\,g\right)$. At a given time instant $k$, assume the optimal control sequence and the corresponding optimal state sequences be represented by $\boldsymbol{\ltvpv{u}}_k^*=\left\{\ltvpv{u}_{k|k}^*,\,\ltvpv{u}_{k+1|k}^*,\,\ldots,\,\ltvpv{u}_{k+N-1|k}^*\right\}$ and $\boldsymbol{\ltvpv{x}}_k^*=\left\{\ltvpv{x}_{k|k}^*,\,\ltvpv{x}_{k+1|k}^*,\,\ldots,\,\ltvpv{x}_{k+N|k}^*\right\}$, respectively. Consider the following candidate control sequence: 
		\begin{subequations}
			\begin{align}
			\ltvpv{u}_{k+1+j|k+1}^0 &= \ltvpv{u}_{k+1+j|k}^* + P_{j}^{g'}w_k + \boldsymbol{m}_{k+1+j|k+1} \nonumber \\ 
			& \hspace{1.8cm} + \Delta {\steady{u}},  \forall j \in \mathbb{Z}_{\left[0:N-2\right]}, \\
			\ltvpv{u}_{k+N|k+1}^0 &= \kappa_f\left(\ltvpv{x}_{k+N|k}^0\right).
			\end{align}\label{eq:cancon}
		\end{subequations}	
		
		The initial condition, state and output dynamics in \eqref{eq:initst}-\eqref{eq:dyn2}, the envelope constraints \eqref{eq:decay}-\eqref{eq:decaycon}, and \eqref{eq:inpcon} (according to \eqref{eq:cancon}) are satisfied by construction. Feasibility at time $k$ satisfies the dynamic constraint \eqref{eq:dyn1}, $\ltvpv{x}_{k+1+j|k}^*=A^{g}\ltvpv{x}^*_{k+j|k}+B^{g}\ltvpv{u}^*_{k+j|k}$. Substituting into system dynamics \eqref{eq:lsys} gives the initial condition, $\ltvpv{x}_{k+1+j|k}= \ltvpv{x}_{k+1+j|k}^* + w_k$, for time $k+1$. Hence, using the initial condition and the candidate control sequence \eqref{eq:cancon}, the candidate state sequence is given as
		\begin{subequations}
			\begin{align}	 
			\ltvpv{x}_{k+1+j|k+1}^0 &= \ltvpv{x}_{k+1+j|k}^* + L_{j}^{g'}w_k + \boldsymbol{n}_{k+1+j|k+1} \nonumber \\ 
			&\hspace{1.5cm} + \boldsymbol{s}_{j} + \Delta {\steady{x}}, \forall j \in \mathbb{Z}_{\left[0:N-1\right]}, \label{eq:ltvcan}\\
			\ltvpv{x}_{k+1+N|k+1}^0 &= A^{g'}\ltvpv{x}_{k+N|k+1}^0 + B^{g'}\ltvpv{u}_{k+N|k+1}^0,
			\end{align}
		\end{subequations}
		
		where $\forall k$, $g$ and $g'\in \left\{\textrm{I},\,\textrm{II},\,\ldots,\,\textrm{XII} \right\}$,
		\begin{subequations}
			\begin{align}
			\boldsymbol{n}_{k+1+j|k+1} & =  \begin{cases}
			\begin{array}{ll}
			\begin{array}{l}
			0,
			\end{array} & j=0\\
			\begin{array}{l}
			\boldsymbol{e}_{k+1}+A^{g'}\boldsymbol{n}_{k+j|k+1}  \\ 
			+B^{g'}\boldsymbol{m}_{k+j|k+1}
			\end{array}, & j>0
			\end{array}\end{cases} \label{eq:n_ltv} \\
			\boldsymbol{m}_{k+1+j|k+1} & =   - K_x\left(\boldsymbol{e}_{k+1}+A^{g'}\boldsymbol{n}_{k+j|k+1}\right) \label{eq:m_ltv} \\
			\boldsymbol{e}_{k+1} & =  
			\left(A^{g'}-A^{g}\right)  \ltvpv{x}^*_{k+1+j|k} \nonumber \\ 
			& \hspace{0.5cm}+\left(B^{g'}-B^{g}\right) \ltvpv{u}^*_{k+1+j|k},  \forall j \in \mathbb{Z}_{\left[0:N-1\right]},
			\end{align}
		\end{subequations}
		\noindent $\mathcal{N}_j= \left\{\boldsymbol{n}_{k+1+j|k+1} | \eqref{eq:n_ltv}\right\} \forall j \in \mathbb{Z}_{\left[0:N-1\right]}$,  $ \mathcal{M}_j = \left\{	\boldsymbol{m}_{k+1+j|k+1} | \eqref{eq:m_ltv} \right\} \forall j \in \mathbb{Z}_{\left[0:N-2\right]}$, and the difference in steady state and input values between successive time steps are represented by $\Delta {\steady{x}} = |{\steady{x}_{k+1}} - {\steady{x}_k}| \in \Delta\mathcal{X}$ and  $\Delta {\steady{u}} = |{\steady{u}_{k+1}} - {\steady{u}_k}| \in \Delta\mathcal{U}$, respectively; $K_x$ is a nilpotent candidate feedback gain. The set $\mathcal{S}_j$ is chosen such that it satisfies $\forall$ $g$ and $g'\in\mathbb{Z}_{\left[1:12\right]}$,
		\begin{align}
		s_j &= \sum_{i=1}^j \left\{\left(A^{g'} L_{j}^{g'} + B^{g'} P_{j}^{g'}\right) - \left(A^{g} L_{j}^{g} + B^{g} P_{j}^{g} \right) \right\}w_k\nonumber \\
		& \hspace{4cm}  \in \mathcal{S}_j\,  \forall k,\,w_k\in \mathcal{W}^{g}.
		\end{align}
		
		Feasibility at $k$ implies $\ltvpv{x}_{k+1+j|k}^*\in{\mathcal{\ltvpv{X}}}_{j+1}^{g},\,\forall j\in\mathbb{Z}_{\left[0:N-2\right]}$. Because of the Pontryagin difference in \eqref{eq:ltv_st}, the state sequence in \eqref{eq:ltvcan} implies $\ltvpv{x}_{k+1+j|k+1}^0\in{\mathcal{\ltvpv{X}}}_{j}^{g'},\,\forall j\in\mathbb{Z}_{\left[0:N-1\right]}$, satisfying \eqref{eq:statecon} at $k+1$. Similarly, feasibility at $k$ also implies $\ltvpv{x}_{k+N|k}^*\in{\mathcal{\ltvpv{X}}_{N}^{g}}$ and hence, \eqref{eq:ltvcan} and \eqref{eq:ltv_termst} imply $\ltvpv{x}_{k+N|k+1}^0\in{\mathcal{X}_{f}}$. Since $\mathcal{X}_f$ is a control invariant set under the controller $\kappa_f$ following from Assumption \ref{ass:kappa}, together with the a nilpotent tightening policy implies $\ltvpv{x}_{k+N|k+1}^0 + L_{N-1}^{g} w_k + \boldsymbol{n}_{k+N|k} + s_{N-1} + \Delta {\steady{x}} \in \mathcal{X}_f$ which combined with the definition \eqref{eq:ltv_termst} implies $\ltvpv{x}_{k+N|k+1}^0\in{\mathcal{\ltvpv{X}}}_N^{g'}$, satisfying the terminal constraint \eqref{eq:termcon} at $k+1$. Furthermore, the choice of $\delta$ in accordance with Remark~\ref{rem:delta} ensures the satisfaction of \eqref{eq:inpslew1}-\eqref{eq:inpslew2} at $k+1$. Therefore, all the constraints are satisfied at $k+1$ with the candidate solution constructed from the optimal solution obtained at time $k$ and hence, the optimisation problem $\mathcal{P}_{N}  \left(x(k+j),\,g'\right)$ is feasible $\forall j>0$. 
	\end{proof}

	\subsection{Offline CT Optimisation}
	\label{sec:ctoffline}
	
	The Pontryagin set difference operations in \eqref{eq:ct_ltv} can be parameterised as affine functions of the disturbance feedback policy, $P_j^g$, to determine the constraint tightening margins for the facets of state and input constraints \cite{Shekhar2012}. However, the sufficiently large size of the empirical disturbance sets obtained at grid points from the system identification data and the lack of control authority to reject the disturbances entered at a time instant $k$ completely before $k+N_{\textrm{np}}$, caused violation of Assumption \ref{ass:non-empty}.	
	
	A non-convex optimisation problem was proposed by \cite{Sankar2017} to estimate the maximal disturbance set that can be handled by the local controller and the corresponding constraint tightening margins to provide robust feasibility guarantees for the disturbances originating from the maximal disturbance set. However, convexity can be recovered by applying convex approximation approaches \cite{Oliveira2000}, \cite{Dinh2010}. For each grid point, the following convex optimisation problem is solved sequentially to potentially obtain an improved locally optimal solution to the nonlinear problem in \cite{Sankar2017}, to determine the maximal disturbance set, ${\mathcal{W}}^{g}_{\max}$, and the respective tightening margins, $\bar{\boldsymbol{{\sigma}}}^{g} \coloneqq \left\{ \ctscp{\sigma}_{0}^{g},\,\ldots,\,\ctscp{\sigma}_{N}^{g}\right\} $ and $\bar{\boldsymbol{{\mu}}}^{g} \coloneqq \left\{ \ctscp{\mu}_{0}^{g},\,\ldots,\,\ctscp{\mu}_{N-1}^{g}\right\}$  with $\ctscp{\sigma}_{0}^g\coloneqq 0_{q\times1}$ and $\ctscp{\mu}_{0}^g\coloneqq 0_{r\times1}$:
	\begin{subequations}
		\begin{align}
		\mathcal{P}_{\textrm{CTSCP}}^{\ctscp{\alpha}}&(g,\,i):\underset{\ctscp{\alpha}^{g},\,\bar{\boldsymbol{{\sigma}}}^{g},\, \bar{\boldsymbol{{\mu}}}^{g}}{\min}\, -\log(\det\left(\ctscp{\alpha})^{g}\right)  \nonumber \\
		&  \hspace{0.75cm} +  \rho \left(\left\|\bar{\boldsymbol{Z}}^{g} - \bar{\boldsymbol{Z}}^{g0}_{i}\right\|_{\max}  + \left\|\bar{{\tilde{\boldsymbol{Z}}}}^{g} - \bar{{\tilde{\boldsymbol{Z}}}}^{g0}_{i}\right\|_{\max} \right.\nonumber\\
		&  \hspace{3.5cm} \left. + \left\|\ctscp{\alpha}^{g} - \ctscp{\alpha}^{g0}_{i}\right\|_{\max}\right)\\
		\textrm{s.t  } & \forall j\in\mathbb{Z}_{\left[0:N-1\right]} \nonumber\\
		& L_{j}^{g}=0_{n\times n},\,\forall j\geq N_{np}\label{eq:nilpot}\\
		& \ctscp{Z}_{j}^{g}\geq0_{a\times q}\label{eq:Z1}\\
		& EL_{j}^{g}=\ctscp{Z}^{{g}^T}_{j}\ctscp{\zeta}\label{eq:dual1}\\
		& 0_{q\times 1}\geq f- E\steady{x}^g-\ctscp{\sigma}_{j}^{g} \label{eq:nonempty1}\\
		&  \ctscp{\sigma_{j+1}}^{g}= \ctscp{\sigma_{j}}^{g}+\left( \ctscp{Z}^{g0^{T}}_{j|i}{\ctscp{\alpha}^{g} }+\ctscp{Z}^{g^{T}}_{j}\ctscp{\alpha}^{g0}_{i}-\ctscp{Z}^{g0^{T}}_{j|i}\ctscp{\alpha}^{g0}_{i}\right)\boldsymbol{1}_a   \label{eq:sigma_cv}\\
		&\forall j\in\mathbb{Z}_{\left[0:N-2\right]} \nonumber\\
		& \bar{\tilde{Z}}_{j}^{g}\geq0_{a\times r}\label{eq:Z2}\\
		& GP_{j}^{g}=\bar{\tilde{Z}}^{g^{T}}_{j}\ctscp{\zeta}\label{eq:dual2}\\
		& 0_{r\times 1}\geq h- G\steady{u}^g-\ctscp{\mu}_{j}^{g}\label{eq:nonempty2}\\
		& L_{j+1}^{g}=A^{g}L_{j}^{g}+B^{g}P_{j}^{g}\label{eq:L}\\
		&  \ctscp{\mu_{j+1}}^{g}= \ctscp{\mu_{j}}^{g}+\left(\bar{\tilde{Z}}^{g0^{T}}_{j|i}{\ctscp{\alpha}^{g} }+{\bar{\tilde{Z}}}^{g^{T}}_{j}\ctscp{\alpha}^{g0}_{i}-\bar{\tilde{Z}}^{g0^{T}}_{j|i}\ctscp{\alpha}^{g0}_{i} \right)\boldsymbol{1}_a \label{eq:mu_cv}\\
		& \ctscp{\alpha}^{g}  \succ 0,
		\end{align} \label{eq:scp}
	\end{subequations}
	
	\noindent where  the dual variables $\ctscp{Z}_j^{g}\in \mathbb{R}^{a\times q}$ and $\bar{\tilde{Z}}^{g}_j\in \mathbb{R}^{a\times r}$ are subject to duality constraints \eqref{eq:dual1} and \eqref{eq:dual2}, and elementwise inequality constraints \eqref{eq:Z1} and \eqref{eq:Z2} (see \cite{Shekhar2012}); $\bar{\boldsymbol{{Z}}}^{g0}_{i}$, $\bar{\tilde{\boldsymbol{{Z}}}}^{g0}_{i}$, $\bar{\boldsymbol{{Z}}}^{g}$ and $\bar{\tilde{\boldsymbol{{Z}}}}^{g}$ denote the horizontal concatenation of the corresponding matrices, i.e., $\bar{\boldsymbol{{Z}}}^{g0}_{i} = \mathscr{C}_{j=0}^{N-1} \ctscp{Z}_{j|i}^{g0^T}$, $\bar{\tilde{\boldsymbol{{Z}}}}^{g0}_{i} = \mathscr{C}_{j=0}^{N-2} \bar{\tilde{{Z}}}_{j|i}^{g0^T}$, $\bar{\boldsymbol{{Z}}}^{g} = \mathscr{C}_{j=0}^{N-1} \ctscp{Z}_j^{g^T}$ and  $\bar{\tilde{\boldsymbol{{Z}}}}^{g} = \mathscr{C}_{j=0}^{N-2} \bar{\tilde{{Z}}}_j^{g^T}$, respectively; $\ctscp{\alpha}^{g0}_{i}$, $\bar{\boldsymbol{{Z}}}^{g0}_{i}$ and $\bar{\tilde{\boldsymbol{{Z}}}}^{g0}_{i}$ are the feasible solution obtained at the iteration $i-1$; $\ctscp{\zeta} = \diag \left\{{I_{a/2},\,-I_{a/2}}\right\}$; $\ctscp{\alpha}^{g} \in \mathbb{R}^{a\times a}$ is a diagonal matrix where the elements scale the facets of the estimated disturbance set such that $\mathcal{\ctscp{W}}^{g}\left(\ctscp{\alpha}^{g}\right)=\left\{ w|\ctscp{\zeta} w\leq\ctscp{\alpha}^{g}\boldsymbol{1}_a\right\}$.
	
	The primary objective of (\ref{eq:scp}) is to maximise the volume of the estimated disturbance set whilst ensuring a non-empty tightened constraint set and nilpotent disturbance feedback policy are obtained for the estimated disturbance set, $\mathcal{\ctscp{W}}^{g}\left(\ctscp{\alpha}^{g}\right)$. The regularisation term scaled by a sufficiently small parameter, $\rho \in \mathbb{R}$, as the secondary objective, penalises the maximum deviation of the elements, $\ctscp{\alpha}^{g}$, $\bar{\boldsymbol{Z}}^{g}$ and $\bar{\tilde{\boldsymbol{Z}}}^{g}$ from the solution obtained at the previous iteration, $\ctscp{\alpha}^{g0}_{i}$, $\bar{\boldsymbol{Z}}^{g0}_{i}$ and $\bar{\tilde{\boldsymbol{Z}}}^{g0}_{i}$, respectively. The constraints (\ref{eq:nonempty1}) and (\ref{eq:nonempty2}) ensure that the tightened state and input constraints remain non-empty along the horizon. The nominal system, $\pv{x}_{k+1|k} =  A^{g}\pv{x}_{k|k}+B^{g}\pv{u}_{k|k}$, is driven to the origin in $N_{np}\leq N$ steps by an $N_{np}$-step nilpotent tightening policy as a result of the constraint (\ref{eq:nilpot}). Furthermore, $\ctscp{\sigma_{N}}^{g} = 0_{q\times 1}$, i.e., no tightening is applied to the positively invariant terminal set, $\mathcal{X}_{f}$, as $L_{N-1}^{g} = 0_{n\times n}$. The tightening margins follow the recursions in \eqref{eq:sigma_cv} and \eqref{eq:mu_cv}.
	
	The following algorithm defines how to solve $\mathcal{P}_{\textrm{CTSCP}}^{\ctscp{\alpha}}(g)$ iteratively to determine a locally optimal solution, $\ctscp{\alpha}^{g*}$, $\bar{\boldsymbol{Z}}^{g*}$, $\bar{\tilde{\boldsymbol{Z}}}^{g*}$, $\bar{\boldsymbol{\sigma}}^{g*}$ and $\bar{\boldsymbol{\mu}}^{g*}$, to the nonlinear problem used in \cite{Sankar2017}. The initial feasible solution, $\ct{\alpha}^{g0}$, $\grave{\boldsymbol{Z}}^{g0}$ and $\grave{\tilde{\boldsymbol{Z}}}^{g0}$, used in the algorithm is obtained by solving the nonlinear maximal disturbance set estimation problem proposed by \cite{Sankar2017}.
	\begin{alg}[Sequential convex program]
		\label{alg:scp}
	\end{alg}
	\begin{algorithmic}
		\State $i\gets 0$
		\State $\ctscp{\alpha}^{g0}_{i}$, $\bar{\boldsymbol{Z}}^{g0}_{i}$ and $\bar{\tilde{\boldsymbol{Z}}}^{g0}_{i} \gets \ct{\alpha}^{g0}$, $\grave{\boldsymbol{Z}}^{g0}$ and $\grave{\tilde{\boldsymbol{Z}}}^{g0}$
		\Repeat
		\State $i \gets i +1 $
		\State Solve \eqref{eq:scp} for $\ctscp{\alpha}^{g}$, $\bar{\boldsymbol{Z}}^{g}$ and $\bar{\tilde{\boldsymbol{Z}}}^{g}$
		\State $\ctscp{\alpha}^{g0}_{i} \gets \ctscp{\alpha}^{g}$
		\State $\bar{\boldsymbol{Z}}^{g0}_{i} \gets \bar{\boldsymbol{Z}}^{g}$
		\State $\bar{\tilde{\boldsymbol{Z}}}^{g0}_{i} \gets \bar{\tilde{\boldsymbol{Z}}}^{g}$
		\Until $\left\|\ctscp{\alpha}^{g0}_{i} - \ctscp{\alpha}^{g0}_{i-1}\right\|_{\max}\leq \delta_{\textrm{tol}}$, $\left\|\bar{\boldsymbol{Z}}^{g0}_{i} - \bar{\boldsymbol{Z}}^{g0}_{i-1}\right\|_{\max}\leq \delta_{\textrm{tol}}$
		and
		$\left\|\bar{\tilde{\boldsymbol{Z}}}^{g0}_{i} - \bar{\tilde{\boldsymbol{Z}}}^{g0}_{i-1}\right\|_{\max}\leq \delta_{\textrm{tol}}$ or $i\leq i_{\max}$
		\If {$i = i_{\max}$}
		\State $\ctscp{\alpha}^{g0}_{i_{\max}}$, $\bar{\boldsymbol{Z}}^{g0}_{i_{\max}}$ and $\bar{\tilde{\boldsymbol{Z}}}^{g0}_{i_{\max}} \gets \ct{\alpha}^{g0}$, $\grave{\boldsymbol{Z}}^{g0}$ and $\grave{\tilde{\boldsymbol{Z}}}^{g0}$
		\EndIf 
		\State $\ctscp{\alpha}^{g*}$, $\bar{\boldsymbol{Z}}^{g*}$ and $\bar{\tilde{\boldsymbol{Z}}}^{g*} \gets \ctscp{\alpha}^{g0}_{i}$, $\bar{\boldsymbol{Z}}^{g0}_{i}$ and $\bar{\tilde{\boldsymbol{Z}}}^{g0}_{i}$.
	\end{algorithmic}
	
	As \eqref{eq:scp} is solved offline, $i_{\max}$ and $\delta_{\textrm{tol}}$ are chosen as $100$ and $10^{-5}$, respectively. The maximal disturbance set for a grid point, ${\mathcal{W}}^{g}_{\max}=\mathcal{\bar{W}}^{g}\left(\ctscp{\alpha}^{g*}\right)$, that can be handled by a given controller at a model grid point, $\left(\omega_e^g,\dot{m}_f^g\right)$, and the corresponding constraint tightening margins, $\bar{\boldsymbol{{\sigma}}}^{g*} \coloneqq \left\{ \ctscp{\sigma}_{0}^{g*},\,\ldots,\,\ctscp{\sigma}_{N}^{g*}\right\} $ and $\bar{\boldsymbol{{\mu}}}^{g*} \coloneqq \left\{ \ctscp{\mu}_{0}^{g*},\,\ldots,\,\ctscp{\mu}_{N-1}^{g*}\right\}$, are determined by solving (\ref{eq:scp}) iteratively according to Algorithm~\ref{alg:scp}. These tightening margins are then used in the online MPC optimisation problem \eqref{eq:mpc} by substituting ${\boldsymbol{{\sigma}}}^{g*}=\bar{\boldsymbol{{\sigma}}}^{g*}$ and ${\boldsymbol{{\mu}}}^{g*}=\bar{\boldsymbol{{\mu}}}^{g*}$.

	\subsection{Controller Calibration}
	As mentioned earlier, each local controller at a grid point has just five tuning parameters and there exists an intuitive correlation between these parameters and the output transient response. To calibrate a local controller, first, for a given choice of tuning parameters, the closed-loop response for a fuelling step or a change in engine speed and fuelling rate is obtained and the type of oscillation observed at each output channel is classified into one of the following:
	\begin{itemize}
		\item Type 0 - The oscillations in the output are suitably small
		\item Type 1 - There is too much undershoot
		\item Type 2 - There is too much overshoot or other non-undershooting oscillation
		\item Type 3 - There are oscillations of both Type 1 and Type 2.
	\end{itemize}
	
	Then for a pair of identified types of oscillations, the calibration table, Table \ref{tab:caltab}, provides heuristic tuning rules/suggestions for selected tuning parameters. If there are multiple rules, then the rules are implemented independently, before applying the possible combinations. Algorithm~\ref{alg:drivecyclecal} defines the calibration procedure followed over a drive cycle.
	
	
	\begin{table}
		\protect\caption{Tuning rules for different types of oscillations observed at the output channels.}
		\centering
		\begin{tabular}{  >{\centering\arraybackslash} m{.1cm} | >{\centering\arraybackslash} m{.75cm} | >{\centering\arraybackslash} m{1.2cm} >{\centering\arraybackslash} m{1.2cm} >{\centering\arraybackslash} m{1.2cm} >{\centering\arraybackslash} m{1.2cm} | }
			\multicolumn{1}{c}{} & \multicolumn{5}{c}{$p_{\textrm{im}}$}\tabularnewline
			\cline{3-6}
			\multicolumn{1}{c}{\parbox[t]{1mm}{\multirow{11}{*}{\rotatebox[origin=c]{90}{$y_{EGR}$}}}} & Type & 0 & 1 & 2 & 3\tabularnewline
			\cline{2-6}
			& 0 & - &  $w^{g} \uparrow$  $\tau_{\boost}^{g} \uparrow$ &  $w^{g} \uparrow$  $\tau_{\boost}^{g} \uparrow$ &  $w^{g} \uparrow$  $\epsilon_{\boost}^{g} \uparrow$  $\tau_{\boost}^{g} \uparrow$\tabularnewline
			\cline{2-6}
			& 1 &  $w^{g} \downarrow$  $\tau_{\EGR}^{g} \uparrow$ &  $\tau_{\boost}^{g} \uparrow$  $\tau_{\EGR}^{g} \uparrow$ &  $\tau_{\EGR}^{g} \uparrow$ &  $\tau_{\EGR}^{g} \uparrow$  $\tau_{\boost}^{g} \uparrow$\tabularnewline
			\cline{2-6}
			& 2 &  $w^{g} \downarrow$  $\tau_{\EGR}^{g} \uparrow$ &  $\tau_{\boost}^{g} \uparrow$ &  $\epsilon_{\boost}^{g} \uparrow$  $\epsilon_{\EGR}^{g} \uparrow$ &  $\tau_{\boost}^{g} \uparrow$  $\tau_{\EGR}^{g} \uparrow$\tabularnewline
			\cline{2-6}
			& 3 &  $w^{g}  \downarrow$  $\epsilon_{\EGR}^{g} \uparrow$  $\tau_{\EGR}^{g} \uparrow$ &  $\tau_{\boost}^{g} \uparrow$  $\tau_{\EGR}^{g} \uparrow$ &  $\tau_{\boost}^{g} \uparrow$  $\tau_{\EGR}^{g} \uparrow$ &  $\tau_{\boost}^{g} \uparrow$  $\tau_{\EGR}^{g} \uparrow$\tabularnewline
			\cline{2-6}
		\end{tabular}
		\label{tab:caltab}
	\end{table}
	
	\begin{alg} [Calibration over drive cycle] 
		\label{alg:drivecyclecal}
	\end{alg}
\begin{enumerate}
	\item Set baseline values for the tuning parameters.
	\item Obtain the closed-loop response over the drive cycle.
	\item Identify the regions along the drive cycle where a better transient response is expected.
	\item The active controllers in the corresponding regions are determined.
	\item The types of oscillations in both output channels are determined from the transient response.
	\item The tuning rules are identified from the calibration table, Table~\ref{tab:caltab}, and applied for each local controller that requires calibration.
	\item Return to step $2$ and repeat until no further improvements are achieved in the output response.
\end{enumerate}

	\section{Simulation Study}
	\label{sec:sim}
	
	The robust controller is implemented in simulations on a nonlinear MVEM of a diesel airpath. The effect of reduced set of calibration parameters on the output transient response is investigated for a step change from `high' to `mid-low' fuelling rate, at `high' engine speed. Therefore, the local controller `X' and the state and input constraint tightening margins, $\boldsymbol{\sigma}^{\textrm{X}*}$ and $ \boldsymbol{\mu}^{\textrm{X}*}$, corresponding to the maximal disturbance set, $\mathcal{W}_{\max}^{\textrm{X}}$, are employed in \eqref{eq:mpc}. The length of the MPC prediction horizon is chosen as $N=7$ and the sampling rate used in this work is consistent with that of the production ECUs.
	
	For the grid point labelled X, the estimated maximal disturbance set arising from Algorithm~\ref{alg:scp} contains approximately $50\%$ of the observed disturbances. It is worth noting that the disturbance set is generated from a large number of data points arising from a pseudo random excitation sequence, and so is likely to be more aggressive than would be encountered in typical engine operation. Furthermore, in practice the high sampling rate was found to aid rapid recovery from a constraint violation.
	
	\begin{rem}
		Multiple aspects of the proposed control architecture helped in achieving a significant prediction horizon length: (i) use of low order model, (ii) avoiding estimators by using physical quantities as states of the linear models, (iii) switching between a family of controllers instead of online linearisations and finally, (iv) use of constraint tightening approach to handle uncertainties due to model mismatch and controller switching, incurs no additional online computational cost.
	\end{rem}
	
	\begin{rem}
		When the steady state input, $\steady{u}_k$, is saturated at a constraint boundary, the size of the maximal disturbance set estimate might be increased by considering disturbances only in certain directions. 
	\end{rem}

	The boost pressure and EGR rate responses for different time constants of the boost pressure envelope are shown in plots (a) and (b) of Fig.~\ref{fig:sim_tau_var}. As expected, lower values of $\tau_{\textrm{boost}}^{{\textrm{X}}}$, for instance $\tau_{\textrm{boost}}^{{\textrm{X}}}=0.08$, encourages the boost pressure response to decay at a faster rate compared to that  obtained for greater values of $\tau_{\textrm{boost}}^{{\textrm{X}}}$. However, lower values of $\tau_{\textrm{boost}}^{{\textrm{X}}}$ will increase the initial height of the boost pressure envelope such that it dominates the primary cost term, resulting in a poor transient response on the EGR rate as seen in Fig.~\ref{fig:sim_tau_var} (b).
	
	\begin{figure}
		\begin{centering}
			\begin{tikzpicture}
			\pgfkeys{/pgf/number format/.cd,fixed,precision=2}
			\node (pic) at (-0,0)	{
				\includegraphics[clip, trim = {0.0cm 0.2cm 0cm 0.5cm}]{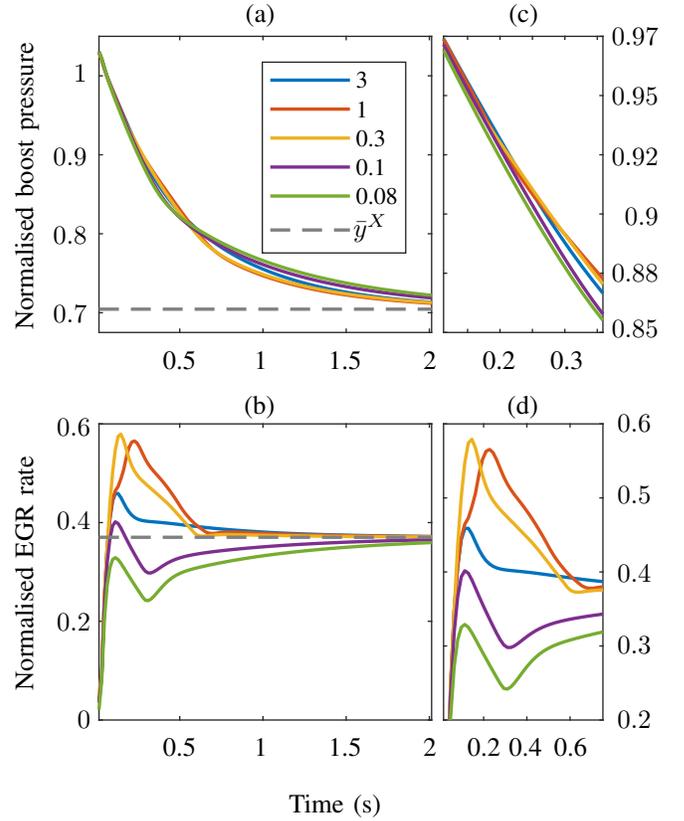}};
			\node [above left = -4mm and -37mm of pic]{(a)};
			\node [ above left = -56mm and -37mm of pic]{(b)};
			\node [above left = -4mm and -72mm of pic]{(c)};
			\node [ above left = -56mm and -72mm of pic]{(d)};
			
			
			\node (rect) [draw =none, fill = white, minimum width = 0.96cm, minimum height = 9.45cm, inner sep = 0pt, below  left = -103mm and -11mm of pic] {};
			
			\node (rect) [draw = none, fill = white, minimum width = 0.55cm, minimum height = 9.45cm, inner sep = 0pt, below  left = -103mm and -86mm of pic] {};
			
			\node [ below  left = -99mm and -12mm of pic] {\pgfmathparse{200/\pimbar}\pgfmathprintnumber\pgfmathresult};
			\node [ below  left = -88mm and -12mm of pic] {\pgfmathparse{180/\pimbar}\pgfmathprintnumber\pgfmathresult};
			\node [ below  left = -77.5mm and -12mm of pic] {\pgfmathparse{160/\pimbar}\pgfmathprintnumber\pgfmathresult};
			\node [ below  left = -67mm and -12mm of pic] {\pgfmathparse{140/\pimbar}\pgfmathprintnumber\pgfmathresult};
			
			\node [ rotate = 90, below left = -102mm and 0.mm of pic] { Normalised boost pressure};

			\node [ below  left = -104mm and -88mm of pic] {\pgfmathparse{195/\pimbar}\pgfmathprintnumber\pgfmathresult};
			\node [ below  left = -96mm and -88mm of pic] {\pgfmathparse{190/\pimbar}\pgfmathprintnumber\pgfmathresult};
			\node [ below  left = -88mm and -88mm of pic] {\pgfmathparse{185/\pimbar}\pgfmathprintnumber\pgfmathresult};
			\node [ below  left = -80mm and -88mm of pic] {\pgfmathparse{180/\pimbar}\pgfmathprintnumber\pgfmathresult};
			\node [ below  left = -72.5mm and -88mm of pic] {\pgfmathparse{175/\pimbar}\pgfmathprintnumber\pgfmathresult};
			\node [ below  left = -65mm and -88mm of pic] {\pgfmathparse{170/\pimbar}\pgfmathprintnumber\pgfmathresult};

			\node [ below  left = -52mm and -12.5mm of pic] {\pgfmathparse{60/\egrbar}\pgfmathprintnumber\pgfmathresult};
			\node [ below  left = -39.5mm and -12.5mm of pic] {\pgfmathparse{40/\egrbar}\pgfmathprintnumber\pgfmathresult};
			\node [ below  left = -26mm and -12.5mm of pic] {\pgfmathparse{20/\egrbar}\pgfmathprintnumber\pgfmathresult};
			\node [ below  left = -13mm and -12.5mm of pic] {\pgfmathparse{0/\egrbar}\pgfmathprintnumber\pgfmathresult};
			
			\node [ rotate = 90, below left = -47mm and -0.mm of pic] { Normalised EGR rate};
			
			\node [ below  left = -52.5mm and -87mm of pic] {\pgfmathparse{60/\egrbar}\pgfmathprintnumber\pgfmathresult};
			\node [ below  left = -43mm and -87mm of pic] {\pgfmathparse{50/\egrbar}\pgfmathprintnumber\pgfmathresult};
			\node [ below  left = -32.5mm and -87mm of pic] {\pgfmathparse{40/\egrbar}\pgfmathprintnumber\pgfmathresult};
			\node [ below  left = -23mm and -87mm of pic] {\pgfmathparse{30/\egrbar}\pgfmathprintnumber\pgfmathresult};
			\node [ below  left = -13mm and -87mm of pic] {\pgfmathparse{20/\egrbar}\pgfmathprintnumber\pgfmathresult};

			
			\node (rect) [draw = none, fill = white, minimum width = 7.0cm, minimum height = 0.35cm, inner sep = 0pt, below  left = -61mm and -80mm of pic] {};
			
			\node [ below  left = -61mm and -26.5mm of pic] {0.5};
			\node [ below  left = -61mm and -36mm of pic] {1};
			\node [ below  left = -61mm and -48.5mm of pic] {1.5};
			\node [ below  left = -61mm and -58.5mm of pic] {2};
			
			\node [ below  left = -61mm and -69mm of pic] {0.2};
			\node [ below  left = -61mm and -78mm of pic] {0.3};
			
			\node (rect) [draw = none, fill = white, minimum width = 7.0cm, minimum height = 0.85cm, inner sep = 0pt, below  left = -9mm and -80mm of pic] {};
			
			\node [ below  left = -9.5mm and -26.5mm of pic] {0.5};
			\node [ below  left = -9.5mm and -36mm of pic] {1};
			\node [ below  left = -9.5mm and -48.5mm of pic] {1.5};
			\node [ below  left = -9.5mm and -58.5mm of pic] {2};
			
			\node [ below  left = -9.5mm and -67mm of pic] {0.2};
			\node [ below  left = -9.5mm and -72.5mm of pic] {0.4};
			\node [ below  left = -9.5mm and -78.5mm of pic] {0.6};
			
			\node [ below   = -2mm  of pic] {Time (s)};
			
				\node (rect) [draw=none , fill = white, minimum width = 0.5cm, minimum height = 0.48cm, inner sep = 0pt, below  left = -78.5mm and -51mm of pic] {$\steady{y}^X$};
			
			\end{tikzpicture}
			\par\end{centering}
		\protect\caption{(a) Boost pressure and (b) EGR rate responses for selected values of  $\tau_{\textrm{boost}}^{{\textrm{X}}}$. Magnified views (c and d).}
		\label{fig:sim_tau_var}
	\end{figure}
	
	\begin{figure}
		\begin{centering}
			\begin{tikzpicture}
			\pgfkeys{/pgf/number format/.cd,fixed,precision=2}
			\node (pic) at (-0,0)	{
				\includegraphics[clip, trim = {0.0cm 0.2cm 0cm 0.5cm}]{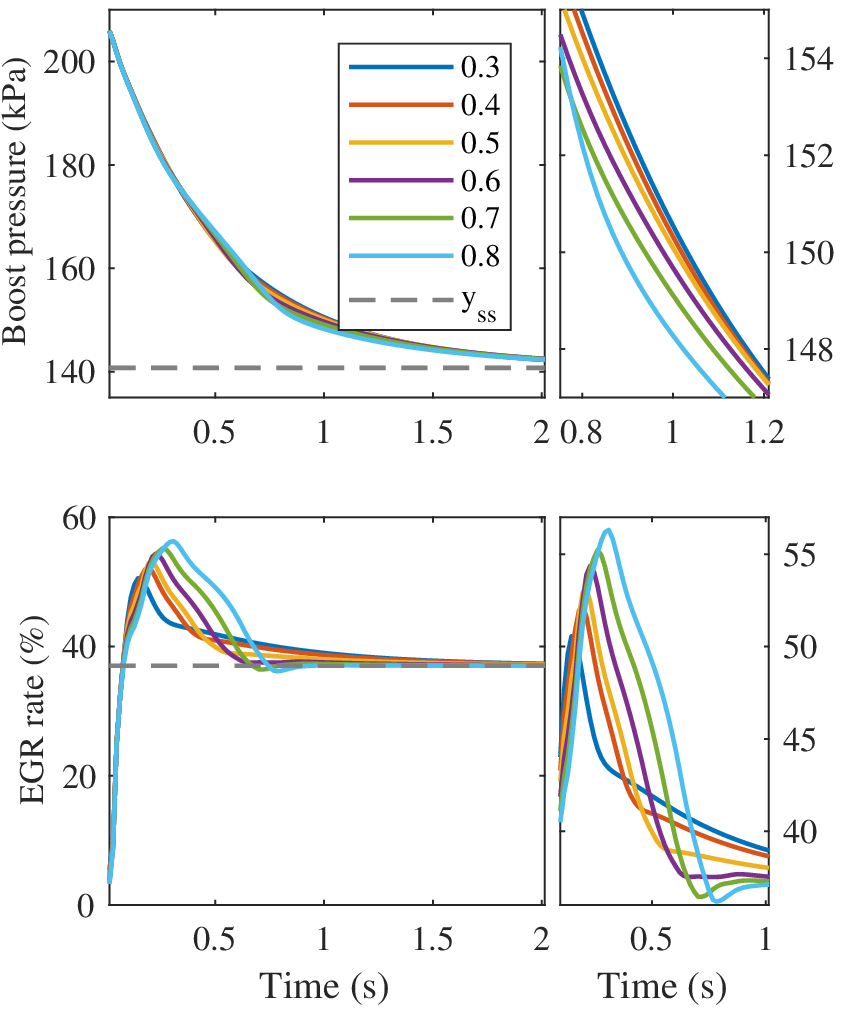}};
			\node [above left = -4mm and -37mm of pic]{(a)};
			\node [ above left = -56mm and -37mm of pic]{(b)};
			\node [above left = -4mm and -72mm of pic]{(c)};
			\node [ above left = -56mm and -72mm of pic]{(d)};
			
			
			\node (rect) [draw =none, fill = white, minimum width = 0.96cm, minimum height = 9.45cm, inner sep = 0pt, below  left = -103mm and -11mm of pic] {};
			
			\node (rect) [draw = none, fill = white, minimum width = 0.55cm, minimum height = 9.45cm, inner sep = 0pt, below  left = -103mm and -86mm of pic] {};
			
			\node [ below  left = -99mm and -12mm of pic] {\pgfmathparse{200/\pimbar}\pgfmathprintnumber\pgfmathresult};
			\node [ below  left = -88mm and -12mm of pic] {\pgfmathparse{180/\pimbar}\pgfmathprintnumber\pgfmathresult};
			\node [ below  left = -77.5mm and -12mm of pic] {\pgfmathparse{160/\pimbar}\pgfmathprintnumber\pgfmathresult};
			\node [ below  left = -67mm and -12mm of pic] {\pgfmathparse{140/\pimbar}\pgfmathprintnumber\pgfmathresult};
			
			\node [ rotate = 90, below left = -102mm and 0.mm of pic] { Normalised boost pressure};

			\node [ below  left = -99mm and -88mm of pic] {\pgfmathparse{154/\pimbar}\pgfmathprintnumber\pgfmathresult};
			\node [ below  left = -89mm and -88mm of pic] {\pgfmathparse{152/\pimbar}\pgfmathprintnumber\pgfmathresult};
			\node [ below  left = -79mm and -88mm of pic] {\pgfmathparse{150/\pimbar}\pgfmathprintnumber\pgfmathresult};
			\node [ below  left = -69mm and -88mm of pic] {\pgfmathparse{148/\pimbar}\pgfmathprintnumber\pgfmathresult};

			\node [ below  left = -52mm and -12.5mm of pic] {\pgfmathparse{60/\egrbar}\pgfmathprintnumber\pgfmathresult};
			\node [ below  left = -39.5mm and -12.5mm of pic] {\pgfmathparse{40/\egrbar}\pgfmathprintnumber\pgfmathresult};
			\node [ below  left = -26mm and -12.5mm of pic] {\pgfmathparse{20/\egrbar}\pgfmathprintnumber\pgfmathresult};
			\node [ below  left = -13mm and -12.5mm of pic] {\pgfmathparse{0/\egrbar}\pgfmathprintnumber\pgfmathresult};
			
			\node [ rotate = 90, below left = -47mm and -0.mm of pic] { Normalised EGR rate};
			
			\node [ below  left = -49mm and -88mm of pic] {\pgfmathparse{55/\egrbar}\pgfmathprintnumber\pgfmathresult};
			\node [ below  left = -39mm and -88mm of pic] {\pgfmathparse{50/\egrbar}\pgfmathprintnumber\pgfmathresult};
			\node [ below  left = -30mm and -88mm of pic] {\pgfmathparse{45/\egrbar}\pgfmathprintnumber\pgfmathresult};
			\node [ below  left = -21mm and -88mm of pic] {\pgfmathparse{40/\egrbar}\pgfmathprintnumber\pgfmathresult};

			
			\node (rect) [draw = none, fill = white, minimum width = 7.2cm, minimum height = 0.35cm, inner sep = 0pt, below  left = -61mm and -83mm of pic] {};
			
			\node (rect) [draw = none, fill = white, minimum width = 7.0cm, minimum height = 0.85cm, inner sep = 0pt, below  left = -9mm and -80mm of pic] {};
			
			\node [ below  left = -61mm and -26.5mm of pic] {0.5};
			\node [ below  left = -61mm and -36mm of pic] {1};
			\node [ below  left = -61mm and -48.5mm of pic] {1.5};
			\node [ below  left = -61mm and -58.5mm of pic] {2};
			
			\node [ below  left = -61mm and -64mm of pic] {0.8};
			\node [ below  left = -61mm and -72mm of pic] {1};
			\node [ below  left = -61mm and -82.5mm of pic] {1.2};

			\node [ below  left = -9.5mm and -26.5mm of pic] {0.5};
			\node [ below  left = -9.5mm and -36mm of pic] {1};
			\node [ below  left = -9.5mm and -48.5mm of pic] {1.5};
			\node [ below  left = -9.5mm and -58.5mm of pic] {2};
			
			\node [ below  left = -9.5mm and -71mm of pic] {0.5};
			\node [ below  left = -9.5mm and -81mm of pic] {1};
			
			\node [ below  = -2mm of pic] {Time (s)};
			
			\node (rect) [draw=none , fill = white, minimum width = 0.5cm, minimum height = 0.48cm, inner sep = 0pt, below  left = -74.5mm and -53mm of pic] {$\steady{y}^X$};
			
			\end{tikzpicture}
			\par\end{centering}
		\protect\caption{(a) Boost pressure and (b) EGR rate responses for selected values of  $w^{{\textrm{X}}}$. Magnified views (c and d).}
		\label{fig:sim_w_var}
	\end{figure}

	\begin{figure}
		\begin{centering}
			\begin{tikzpicture}
			\pgfkeys{/pgf/number format/.cd,fixed,precision=2}
			
			\node (pic) at (0,0)
			{\includegraphics[clip, trim = {0.5cm 0.35cm 0cm 0cm}]{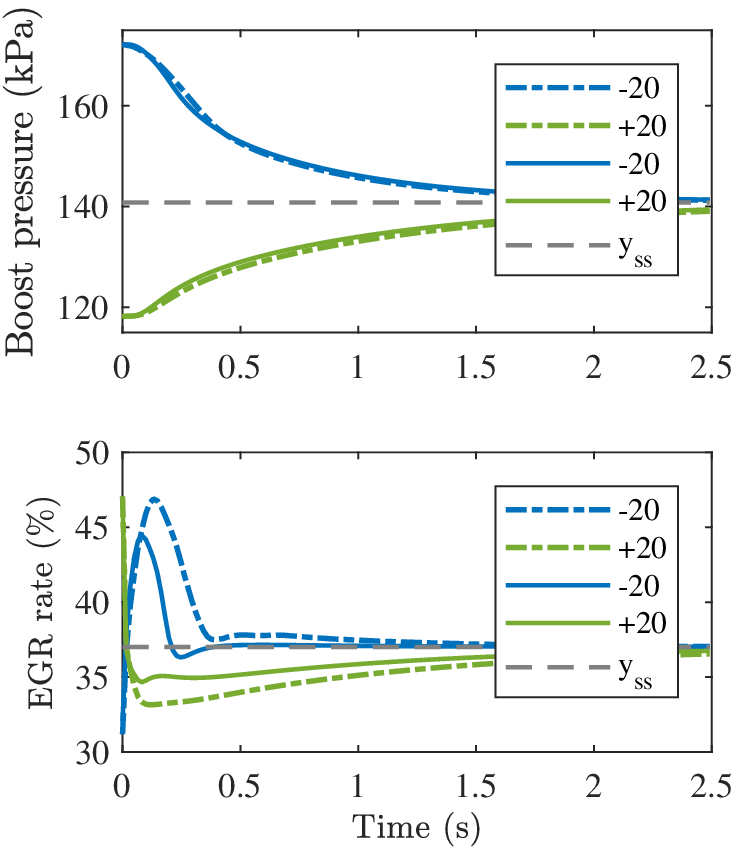}};
			
			\node [above = -8mm of pic]{(a)};
			\node [ above = -51mm of pic]{(b)};
			
			
			\node (rect) [draw =none , fill = white, minimum width = 0.9cm, minimum height = 7.8cm, inner sep = 0pt, below  left = -83mm and -8mm of pic] {};
			
			\node [ below  left = -76mm and -8mm of pic] {\pgfmathparse{160/\pimbar}\pgfmathprintnumber\pgfmathresult};
			\node [ below  left = -66mm and -8mm of pic] {\pgfmathparse{140/\pimbar}\pgfmathprintnumber\pgfmathresult};
			\node [ below  left = -55mm and -8mm of pic] {\pgfmathparse{120/\pimbar}\pgfmathprintnumber\pgfmathresult};
			
			\node [ rotate = 90, below left = -84mm and 4.mm of pic] { Normalised boost pressure};
			
			\node [ below  left = -41mm and -8mm of pic] {\pgfmathparse{50/\egrbar}\pgfmathprintnumber\pgfmathresult};
			\node [ below  left = -33mm and -8mm of pic] {\pgfmathparse{45/\egrbar}\pgfmathprintnumber\pgfmathresult};
			\node [ below  left = -25.5mm and -8mm of pic] {\pgfmathparse{40/\egrbar}\pgfmathprintnumber\pgfmathresult};
			\node [ below  left = -18mm and -8mm of pic] {\pgfmathparse{35/\egrbar}\pgfmathprintnumber\pgfmathresult};
			\node [ below  left = -10mm and -8mm of pic] {\pgfmathparse{30/\egrbar}\pgfmathprintnumber\pgfmathresult};
			
			\node [ rotate = 90, below left = -39.5mm and 4.mm of pic] { Normalised EGR rate};
			
			
			\node (rect) [draw =none, fill = white, minimum width = 6.3cm, minimum height = 0.35cm, inner sep = 0pt, below  left = -49.5mm and -71mm of pic] {};
			
			\node (rect) [draw =none, fill = white, minimum width = 6.3cm, minimum height = 0.65cm, inner sep = 0pt, below  left = -7mm and -71mm of pic] {};
			
			\node [ below  left = -49.5mm and -11mm of pic] {0};
			\node [ below  left = -49.5mm and -24mm of pic] {0.5};
			\node [ below  left = -49.5mm and -34.5mm of pic] {1};
			\node [ below  left = -49.5mm and -48.5mm of pic] {1.5};
			\node [ below  left = -49.5mm and -58.5mm of pic] {2};
			\node [ below  left = -49.5mm and -72.5mm of pic] {2.5};
			
			\node [ below  left = -7mm and -11mm of pic] {0};
			\node [ below  left = -7mm and -24mm of pic] {0.5};
			\node [ below  left = -7mm and -34.5mm of pic] {1};
			\node [ below  left = -7mm and -48.5mm of pic] {1.5};
			\node [ below  left = -7mm and -58.5mm of pic] {2};
			\node [ below  left = -7mm and -72.5mm of pic] {2.5};
			
			\node [ below   = -2mm of pic] {Time (s)};
			
			\node (rect) [draw=none , fill = white, minimum width = 0.5cm, minimum height = 0.48cm, inner sep = 0pt, below  left = -61.5mm and -64mm of pic] {$\steady{y}^X$};
			
				\node (rect) [draw=none , fill = white, minimum width = 0.5cm, minimum height = 0.48cm, inner sep = 0pt, below  left = -19mm and -64mm of pic] {$\steady{y}^X$};

			\end{tikzpicture}
			\par\end{centering}
		\protect\caption{(a) Boost pressure and (b) EGR rate responses for the baseline (dash-dotted line) and final (solid line) calibration parameter settings.}
		\label{fig:sim_step_resp}
	\end{figure}
	
	The effect of the parameter $w^{{\textrm{X}}}$ on the two output channels is shown in plots (a) and (b) Fig.~\ref{fig:sim_w_var}. From \ref{fig:sim_w_var} (c), it can be seen that the magnitude of overshoot in the EGR rate response increases and the boost pressure response decays faster as the value of $w^{{\textrm{X}}}$ is increased greater than $0.5$, since when $w^{{\textrm{X}}}>0.5$, minimisation of the boost pressure envelope is given priority over minimising the EGR rate envelope and vice versa. The discussion on the effect of the smoothness parameters, $\epsilon_{\boost}^{{\textrm{X}}}$ and $\epsilon_{\EGR}^{{\textrm{X}}}$, is omitted for brevity. 
	
	The output response for fuelling rate step changes of magnitude $\pm\SI{20}{mm^3/stroke}$ about the grid point X is shown in Fig.~\ref{fig:sim_step_resp} for baseline and final choice of calibration parameters. The baseline parameters are chosen as: $\tau_{\textrm{boost}}^{{\textrm{X}}}=1,\,  \tau_{\textrm{EGR}}^{{\textrm{X}}} = 1,\, w^{{\textrm{X}}}=0.5,\, $ $\epsilon_{\textrm{boost}}^{{\textrm{X}}} = 0 \textrm{ and } \epsilon_{\textrm{EGR}}^{{\textrm{X}}} = 0$. A Type 0 oscillation is observed in the boost pressure channel, whereas, an overshooting response is observed in the EGR rate for both fuel steps, indicating a Type 1 oscillation. The controller is calibrated by using the calibration table, Table \ref{tab:caltab}, to reduce the peak overshoot in the EGR rate, whilst not adversely affecting the boost pressure response. As shown in Fig.~\ref{fig:sim_step_resp}, the magnitude of the overshoots are reduced in the EGR rate response with the final choice calibration parameters - $\tau_{\textrm{boost}}^{{\textrm{X}}}=1,\,  \tau_{\textrm{EGR}}^{{\textrm{X}}} = 5,\, w^{{\textrm{X}}}=0.4,\, $ $\epsilon_{\textrm{boost}}^{{\textrm{X}}} = 0 \textrm{ and } \epsilon_{\textrm{EGR}}^{{\textrm{X}}} = 0$. Increasing $\tau_{\textrm{EGR}}^{{\textrm{X}}} > 5$, will reduce the overshoot in the EGR rate, however, will result in a sluggish boost pressure response, which is undesirable. Similarly, decreasing $w^{{\textrm{X}}}<0.4$ worsens the boost response and hence, the calibration is finished. Extensive simulation studies (which are not presented in the paper for brevity) conducted at different operating points over several fuelling steps indicate that the calibration rules suggested in Table \ref{tab:caltab} produce the desired consequence at the output responses.
	
	
	\section{  Experimental Results}
	\label{sec:expres}
	
	In this section, the experimental results obtained by implementing the proposed MPC on a diesel engine bench and calibrating for fuelling step changes about a steady operating condition and over drive cycles are presented. 
	
	\subsection{Real Time Implementation}
	\label{sec:RT_implementation}
	A test bench at Toyota's Higashi-Fuji Technical Center in Susono, Japan is used to experimentally demonstrate the calibration efficacy of the proposed controller. The test bench is equipped with a diesel engine and a transient dynamometer. A d\textsc{SPACE} DS1006 real-time processor board \cite{DS1006} is used to implement the control system described in Section~\ref{sec:mpc}. A block diagram of the controller configuration on the test bench is shown in Fig.~\ref{fig:con_config}. The ECU logs sensor data from the engine and transmits the current state information to the controller. Also, the ECU directly controls all engine sub-systems. However, the ECU commands for the three actuators - throttle, EGR valve and VGT, can be overridden with the MPC commands through enabling a virtual switch shown in Fig.~\ref{fig:con_config} from the ControlDesk interface. For the current engine speed and fuelling rate, $\left(\omega_e,\dot{m}_f\right)$, the model, tuning parameters and CT margins are selected based on the switched LTI-MPC strategy and used by the MPC at each time instant as shown in Fig.~\ref{fig:con_config}.
	
	\begin{figure}
		\begin{centering}
			\begin{tikzpicture} [scale = 0.7, transform shape]
			
			\node [text width=6em, minimum height=3em, text centered,draw,fill=gray!20, inner sep = 0pt] (engine) at (0,0) {Diesel \\ engine};
			\node [text width=6.5em, minimum height=3em, text centered,draw,fill=gray!20, right = 7mm of engine, inner sep = 0pt] (dyn)  {Transient dynamometer};
			\node [text width=6.5em, minimum height=3em, text centered,draw,fill=gray!20, above = 7mm of dyn, inner sep = 0pt] (dyncon) {Dynamometer controller};
			\node [text width=6.0em, minimum height=3.0em, text centered,draw, inner sep = 0pt, above = 7mm of engine] (omegaesetpoint) {Engine speed setpoint};
			
			\draw [line width=3pt] (engine) edge (dyn);
			\draw [latex-latex,thick] (dyn) edge (dyncon);
			\draw [-latex] (omegaesetpoint) edge (dyncon);
			
			\node [draw, fill=gray!60, below left = -6.20mm and 6mm of engine, minimum width = 0.8cm, minimum height = 0.9cm, inner sep = 0pt] (switch) {};
			\node [draw, circle, minimum size = 1mm, below left = 6.7mm and 2mm of switch.north, inner sep = 0pt] (inportmpc) {};
			\node [draw, circle, minimum size = 1mm, below left = 2.mm and 2mm of switch.north, inner sep = 0pt] (inportecu) {};
			\node [draw, circle, minimum size = 1mm, below left = 4.1mm and -2.5mm of switch.north, inner sep = 0pt] (outport) {};
			\node [text width=4em,text centered, below = 1mm of switch, inner sep = 0pt] (swtichtextbox) {};
			\draw (inportmpc) -- (outport);

			\node [text width=3em, minimum height=3em, text centered,draw,fill=gray!20, above left = -11mm and 23mm of engine] (ecu)  {ECU};
			\node [text width=3.0em, minimum height=3em, text centered,draw, left = 23mm of omegaesetpoint] (load) {Load};
			\node [text width=3em, minimum height=3em, text centered,draw,fill=gray!20, below = 7mm of ecu] (mpc) {MPC};
			\draw [-latex] (load) -- (ecu);
			
			\path (ecu.east) -- (ecu.north east) coordinate[pos=0.56] (ecu1);
			\path (ecu.east) -- (ecu.south east) coordinate[pos=0.23] (ecu2);
			
			\path (engine.west) -- (engine.north west) coordinate[pos=0.501] (eng1);
			\path (engine.west) -- (engine.south west) coordinate[pos=0.682] (eng2);
			
			\draw (ecu2) -- (inportecu);
			\draw [latex-latex] (ecu1) -- (eng1);
			\draw [-latex] (ecu) -- (mpc) node [pos = 0.6, right] {$x_k,$ $y_k$};
			\draw [-latex] (outport) -- (eng2)  node [pos = 0.7, above] {$u$};

			\node [text width=5em, minimum height=3em, text centered,draw,fill=gray!20, above left = -0mm and 9mm of ecu] (modpar) {Model, tuning parameters and CT margins};
			\node [text width=5em, minimum height=3em, text centered,draw,fill=gray!20, above left = -24mm and 9mm of ecu] (maps) {Set point maps};
			\node [text width=4em,text centered, above left = -5mm and 1mm of ecu, inner sep = 0pt] {$\omega_e,$ $\dot{m}_f$};
			\draw [-latex] (ecu) -- ++(-2.5,0) -| (maps) ;
			\draw [-latex] (ecu) -- ++(-2.5,0) -| (modpar) ;
			\draw [-latex] (maps.east) -- ++(0.5,0) |- (mpc.146); 
			\draw [-latex] (modpar.west) -- ++(-0.3,0) |- (mpc.210) ;
			\node [text width=10em, minimum height=2em, above left = -20mm and -1mm of mpc, inner sep = 0pt] {$A^{g},$ $B^{g},$ $C^{g},$ $D^{g},$ $\tau_{\boost}^{g},$ $\tau_{\EGR}^{g},$ $w^{g},$ $\boldsymbol{\epsilon}^{g},$ $\boldsymbol{\sigma}^{g*},$ $\boldsymbol{\mu}^{g*}$};
			\node [text width=2em, minimum height=3em, above left = -5mm and 0mm of mpc, inner sep = 0pt] {$\steady{u}_{k},$ $\steady{x}_{k},$ $\steady{y}_{k}$};
			\path [draw] (mpc.east) -- ++(0.5,0)  |- (inportmpc) ;
			\end{tikzpicture}
			\par\end{centering}
		\protect\caption{Controller configuration.}
		\label{fig:con_config}
	\end{figure}
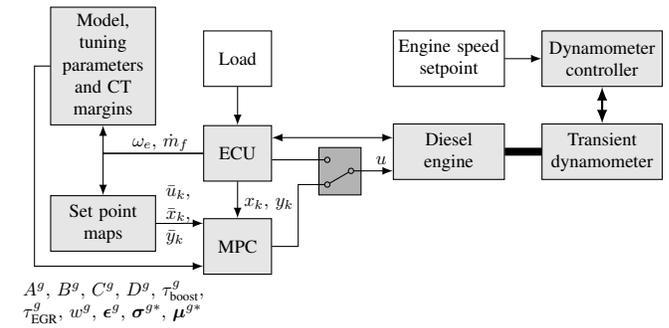
	
	For real time implementation of the controller on the dSPACE platform, it is necessary to carefully choose a QP solver designed to run fast on embedded hardware. The quadratic programming in C (QPC) solver suite \cite{Wills2012} is chosen for this purpose. In particular, the interior-point solver \texttt{qpip} is used for solving the MPC optimisation problem \eqref{eq:mpc}, in conjunction with \textsc{Matlab} R2010b \& Simulink Real-Time Workshop and d\text{SPACE} RTI \& HIL Software v7.4. During steady state calibration, the dynamometer controller maintains the engine at the desired speed while the changes in the fuelling rate are implemented through the dSPACE ControlDesk interface.

	\begin{figure}
		\begin{centering}
			\begin{tikzpicture}
			\pgfkeys{/pgf/number format/.cd,fixed,precision=2}
			
			\node (pic) at (0,0)
			{\includegraphics[clip, trim = {0.5cm 0.35cm 0cm 0cm}]{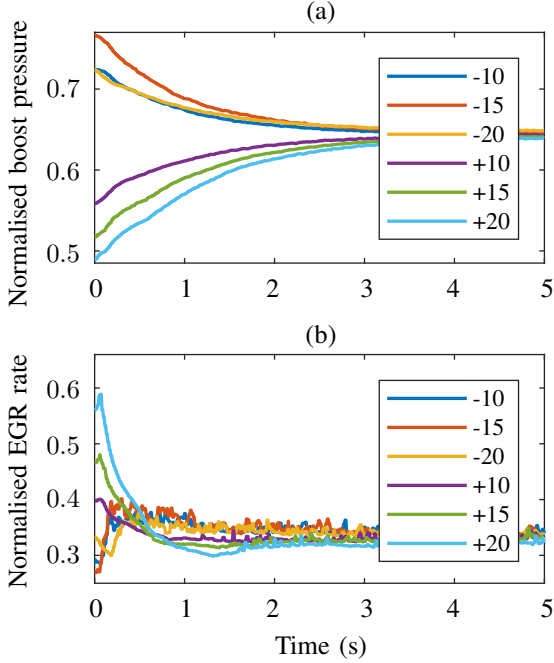}};
			
			\node [above = -8mm of pic]{(a)};
			\node [ above = -51mm of pic]{(b)};
			
			
			\node (rect) [draw =none , fill = white, minimum width = 0.9cm, minimum height = 7.8cm, inner sep = 0pt, below  left = -83mm and -8mm of pic] {};
			
			\node [ below  left = -76mm and -8mm of pic] {\pgfmathparse{140/\pimbar}\pgfmathprintnumber\pgfmathresult};
			\node [ below  left = -65mm and -8mm of pic] {\pgfmathparse{120/\pimbar}\pgfmathprintnumber\pgfmathresult};
			\node [ below  left = -54mm and -8mm of pic] {\pgfmathparse{100/\pimbar}\pgfmathprintnumber\pgfmathresult};
			
			\node [ rotate = 90, below left = -84mm and 4.mm of pic] { Normalised boost pressure};
			
			\node [ below  left = -36mm and -8mm of pic] {\pgfmathparse{60/\egrbar}\pgfmathprintnumber\pgfmathresult};
			\node [ below  left = -28.5mm and -8mm of pic] {\pgfmathparse{50/\egrbar}\pgfmathprintnumber\pgfmathresult};
			\node [ below  left = -21.5mm and -8mm of pic] {\pgfmathparse{40/\egrbar}\pgfmathprintnumber\pgfmathresult};
			\node [ below  left = -14mm and -8mm of pic] {\pgfmathparse{30/\egrbar}\pgfmathprintnumber\pgfmathresult};
			
			\node [ rotate = 90, below left = -39.5mm and 4.mm of pic] { Normalised EGR rate};
			
			
			\node (rect) [draw=none, fill = white, minimum width = 6.3cm, minimum height = 0.38cm, inner sep = 0pt, below  left = -49.5mm and -71mm of pic] {};
			
			\node (rect) [draw =none, fill = white, minimum width = 6.3cm, minimum height = 0.65cm, inner sep = 0pt, below  left = -7mm and -70mm of pic] {};
			
			\node [ below  left = -49.5mm and -11mm of pic] {0};
			\node [ below  left = -49.5mm and -23mm of pic] {1};
			\node [ below  left = -49.5mm and -34.5mm of pic] {2};
			\node [ below  left = -49.5mm and -47mm of pic] {3};
			\node [ below  left = -49.5mm and -58.5mm of pic] {4};
			\node [ below  left = -49.5mm and -71mm of pic] {5};
			
			\node [ below  left = -7mm and -11mm of pic] {0};
			\node [ below  left = -7mm and -23mm of pic] {1};
			\node [ below  left = -7mm and -34.5mm of pic] {2};
			\node [ below  left = -7mm and -47mm of pic] {3};
			\node [ below  left = -7mm and -58.5mm of pic] {4};
			\node [ below  left = -7mm and -71mm of pic] {5};
			
			\node [ below  = -2mm of pic] {Time (s)};

			\end{tikzpicture}
			\par\end{centering}
		\protect\caption{(a) Boost pressure and (b) EGR rate responses for the baseline calibration parameter setting: $\tau_{\textrm{boost}}^{{\textrm{VI}}}=1,\,  \tau_{\textrm{EGR}}^{{\textrm{VI}}} = 1,\, w^{{\textrm{VI}}}=0.5,\, $ $\epsilon_{\textrm{boost}}^{{\textrm{VI}}} = 0 \textrm{ and } \epsilon_{\textrm{EGR}}^{{\textrm{VI}}} = 0$.}
		\label{fig__ss_base}
	\end{figure}

	\begin{figure}
		\begin{centering}
			\begin{tikzpicture}
			\pgfkeys{/pgf/number format/.cd,fixed,precision=2}
			
			\node (pic) at (0,0)
			{\includegraphics[clip, trim = {0.5cm 0.35cm 0cm 0cm}]{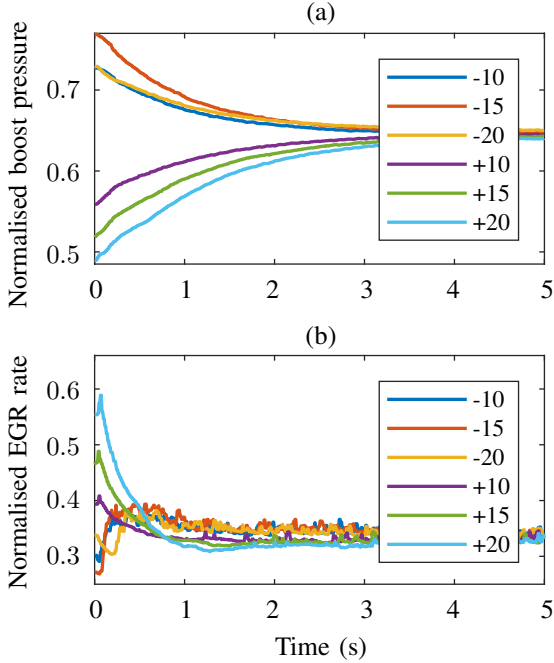}};
			
			\node [above = -8mm of pic]{(a)};
			\node [ above = -51mm of pic]{(b)};
			
			
			\node (rect) [draw =none , fill = white, minimum width = 0.9cm, minimum height = 7.8cm, inner sep = 0pt, below  left = -83mm and -8mm of pic] {};
			
			\node [ below  left = -76mm and -8mm of pic] {\pgfmathparse{140/\pimbar}\pgfmathprintnumber\pgfmathresult};
			\node [ below  left = -65mm and -8mm of pic] {\pgfmathparse{120/\pimbar}\pgfmathprintnumber\pgfmathresult};
			\node [ below  left = -54mm and -8mm of pic] {\pgfmathparse{100/\pimbar}\pgfmathprintnumber\pgfmathresult};
			
			\node [ rotate = 90, below left = -84mm and 4.mm of pic] { Normalised boost pressure};
			
			\node [ below  left = -36mm and -8mm of pic] {\pgfmathparse{60/\egrbar}\pgfmathprintnumber\pgfmathresult};
			\node [ below  left = -28.5mm and -8mm of pic] {\pgfmathparse{50/\egrbar}\pgfmathprintnumber\pgfmathresult};
			\node [ below  left = -21.5mm and -8mm of pic] {\pgfmathparse{40/\egrbar}\pgfmathprintnumber\pgfmathresult};
			\node [ below  left = -14mm and -8mm of pic] {\pgfmathparse{30/\egrbar}\pgfmathprintnumber\pgfmathresult};
			
			\node [ rotate = 90, below left = -39.5mm and 4.mm of pic] { Normalised EGR rate};
			
			
			\node (rect) [draw=none, fill = white, minimum width = 6.3cm, minimum height = 0.38cm, inner sep = 0pt, below  left = -49.5mm and -71mm of pic] {};
			
			\node (rect) [draw =none, fill = white, minimum width = 6.3cm, minimum height = 0.65cm, inner sep = 0pt, below  left = -7mm and -70mm of pic] {};
			
			\node [ below  left = -49.5mm and -11mm of pic] {0};
			\node [ below  left = -49.5mm and -23mm of pic] {1};
			\node [ below  left = -49.5mm and -34.5mm of pic] {2};
			\node [ below  left = -49.5mm and -47mm of pic] {3};
			\node [ below  left = -49.5mm and -58.5mm of pic] {4};
			\node [ below  left = -49.5mm and -71mm of pic] {5};
			
			\node [ below  left = -7mm and -11mm of pic] {0};
			\node [ below  left = -7mm and -23mm of pic] {1};
			\node [ below  left = -7mm and -34.5mm of pic] {2};
			\node [ below  left = -7mm and -47mm of pic] {3};
			\node [ below  left = -7mm and -58.5mm of pic] {4};
			\node [ below  left = -7mm and -71mm of pic] {5};
			
			\node [ below  = -2mm of pic] {Time (s)};

			\end{tikzpicture}
			\par\end{centering}
		\protect\caption{(a) Boost pressure and (b) EGR rate responses for the final parameter setting: $\tau_{\textrm{boost}}^{{\textrm{VI}}}=1,\,  \tau_{\textrm{EGR}}^{{\textrm{VI}}} = 1,\, w^{{\textrm{VI}}}=0.4,\, $ $\epsilon_{\textrm{boost}}^{{\textrm{VI}}} = 0 \textrm{ and } \epsilon_{\textrm{EGR}}^{{\textrm{VI}}} = 0$.}
		\label{fig__ss_tr1}
	\end{figure}

	\subsection{Calibration at a Steady State Condition}
	The controller calibration based on the response of the closed-loop system and the calibration table, for step changes in the fuelling rate about the steady state  linearisation/model grid point VI is presented. The fuelling steps considered in this study are $\pm 10,\, \pm 15$ and $\pm\SI{20}{mm^3/stroke}$. 
	
	The performance of the controller for step changes in the fuelling rate with baseline settings of the tuning parameters is shown in Fig.~\ref{fig__ss_base}.  The baseline tuning parameters are chosen as:  $\tau_{\textrm{boost}}^{{\textrm{VI}}} = 1,\, \tau_{\textrm{EGR}}^{{\textrm{VI}}} = 1,\, w^{{\textrm{VI}}}= 0.5,\, \epsilon_{\textrm{boost}}^{{\textrm{VI}}} = 0 \textrm{ and } \epsilon_{\textrm{EGR}}^{{\textrm{VI}}} = 0$. The fixed cost function parameters in \eqref{eq:mpccost} are chosen as $\gamma = \SI{2e-4}{}$, $R = I_3$, $w_{\boost} = \SI{40}{kPa}$ and $w_{\EGR} = 0.6$. The horizon length and the controller sampling time are identical to those used in the simulation study. From the Fig.~\ref{fig__ss_base} (a), a smooth response is noticed in the transients obtained for the intake manifold pressure indicating a Type 0 oscillation. However, an undershooting followed by overshooting behaviour is obtained for the EGR rate as seen in Fig.~\ref{fig__ss_base} (b). Therefore, the type of oscillation in EGR rate output channel is identified as Type 3. By using the reduced set of tuning parameters and the calibration table, the controller is calibrated to diminish the oscillatory behaviour in the EGR rate response whilst not adversely affecting the intake manifold pressure response.	
	
	Since there are three suggestions in Table~\ref{tab:caltab} for the identified types of oscillations - Type 0 and Type 3 oscillations in boost and EGR rate response, respectively, the first suggestion to decrease $w^{{\textrm{VI}}}$ is implemented. By setting $w^{{\textrm{VI}}}=0.4$, the closed-loop response of the controller is analysed again. The magnitude of overshoot in the EGR rate responses is reduced as seen in Fig.~\ref{fig__ss_tr1}. However, the type of oscillation in EGR rate channel remains as Type 3 with Type 0 in boost pressure channel. As the first rule has been tested, the next tuning rule from the Table~\ref{tab:caltab} - increase $\epsilon_{\textrm{EGR}}^{{\textrm{VI}}}$, is implemented and the closed-loop is checked for improvements in the transient response. A similar EGR rate response is observed with the new rule. Finally, as per the third rule from the calibration table, $\tau_{\textrm{EGR}}^{{\textrm{VI}}}$ is increased to $3$ and the closed-loop response is obtained. Once again, the output responses obtained with the third rule is similar to the response obtained by applying first rule. Implementing multiple tuning rules to reduce the overshoot in the EGR rate response adversely affected the performance by slowing down the boost response. Therefore, $\tau_{\textrm{boost}}^{{\textrm{VI}}}=1,\,  \tau_{\textrm{EGR}}^{{\textrm{VI}}} = 1,\,  w^{{\textrm{VI}}}=0.4,\, \epsilon_{\textrm{boost}}^{{\textrm{VI}}} = 0 \textrm{ and } \epsilon_{\textrm{EGR}}^{{\textrm{VI}}} = 0$ are chosen as the final values for the tuning parameters.

	\subsection {Calibration over Drive Cycles}
	
	The performance of the switched LTI-MPC architecture and the procedure followed for tuning the family of twelve calibration-friendly local controllers based on Algorithm \ref{alg:drivecyclecal}, over the extra-urban driving cycle (EUDC) and the medium phase of the worldwide harmonised light vehicle test procedure (WLTP) will be discussed.

	\subsubsection{Calibration over EUDC}
	The controllers with baseline parameter setting are tested over EUDC and the tracking performance is shown in Fig.~\ref{fig:eudc_base}. The highlighted regions along the drive cycle in Fig.~\ref{fig:eudc_base} have undesired transient responses in one or both output channels as listed in Table \ref{tab:eudc_base}. For example, in the region $B_e$ in Fig.~\ref{fig:eudc_base}, an oscillatory behaviour is observed in the EGR rate response whilst the boost pressure response is acceptable. Hence, the types of oscillations in the output channels are identified as $\textrm{Type}\, 0$ and $\textrm{Type}\,3$, respectively, will be represented as $\left(0,\,3\right)$.
	
	\begin{figure}
		\begin{centering}
			\begin{tikzpicture}
			\pgfkeys{/pgf/number format/.cd,fixed,precision=2}
			
			\node (pic) at (0,0)
			{\includegraphics[clip, trim = {0.7cm 1.4cm 1.3cm 1.2cm},width = 8.7cm]{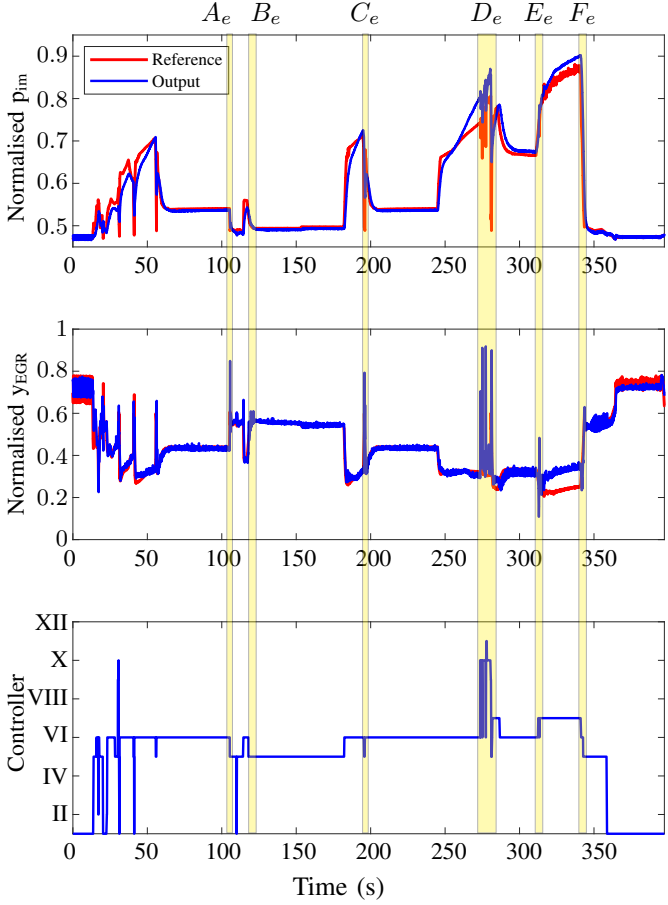}};
			
			
			\node (rect) [draw = none , fill = white, minimum width = 0.76cm, minimum height = 11cm, inner sep = 0pt, below  left = -113mm and -9mm of pic] {};
			
						\node (rect) [draw = none, fill = white, minimum width = 8.1cm, minimum height = 0.3cm, inner sep = 0pt, below  left = -83mm and -89mm of pic] {};
			
			\node (rect) [draw = none, fill = white, minimum width = 8.1cm, minimum height = 0.3cm, inner sep = 0pt, below  left = -44mm and -89mm of pic] {};
			
			\node (rect) [draw = none, fill = white, minimum width = 8.1cm, minimum height = 0.3cm, inner sep = 0pt, below  left = -5mm and -89mm of pic] {};
			
			\node [ below  left = -111.5mm and -10mm of pic] {\small{\pgfmathparse{180/\pimbar}\pgfmathprintnumber\pgfmathresult}};
			\node [ below  left = -105.5mm and -10mm of pic] {\small{\pgfmathparse{160/\pimbar}\pgfmathprintnumber\pgfmathresult}};
			\node [ below  left = -100mm and -10mm of pic] {\small{\pgfmathparse{140/\pimbar}\pgfmathprintnumber\pgfmathresult}};
			\node [ below  left = -94.5mm and -10mm of pic] {\small{\pgfmathparse{120/\pimbar}\pgfmathprintnumber\pgfmathresult}};
			\node [ below  left = -88.5mm and -10mm of pic] {\small{\pgfmathparse{100/\pimbar}\pgfmathprintnumber\pgfmathresult}};
			
			\node [ rotate = 90, below left = -108mm and 0.5mm of pic] {\small{Normalised p$_{\textrm{im}}$}};
			
			\node [ below  left = -75mm and -10mm of pic] {\small{\pgfmathparse{100/\egrbar}\pgfmathprintnumber\pgfmathresult}};
			\node [ below  left = -69.5mm and -10mm of pic] {\small{\pgfmathparse{80/\egrbar}\pgfmathprintnumber\pgfmathresult}};
			\node [ below  left = -63.5mm and -10mm of pic] {\small{\pgfmathparse{60/\egrbar}\pgfmathprintnumber\pgfmathresult}};
			\node [ below  left = -58mm and -10mm of pic] {\small{\pgfmathparse{40/\egrbar}\pgfmathprintnumber\pgfmathresult}};
			\node [ below  left = -52.5mm and -10mm of pic] {\small{\pgfmathparse{20/\egrbar}\pgfmathprintnumber\pgfmathresult}};
			\node [ below  left = -47mm and -10mm of pic] {\small{\pgfmathparse{0/\egrbar}\pgfmathprintnumber\pgfmathresult}};
			
			\node [ rotate = 90, below left = -71mm and 0.5mm of pic] { \small{Normalised {y$_{\textrm{EGR}}$}}};
			
			\node [ below  left = -36mm and -10mm of pic] {\small{XII}};
			\node [ below  left = -31mm and -10mm of pic] {\small{X}};
			\node [ below  left = -26mm and -10mm of pic] {\small{VIII}};
			\node [ below  left = -21mm and -10mm of pic] {\small{VI}};
			\node [ below  left = -15.5mm and -10mm of pic] {\small{IV}};
			\node [ below  left = -10.5mm and -10mm of pic] {\small{II}};
			
			\node [ rotate = 90, below left = -28mm and 0.5mm of pic] { \small{Controller}};


			\node [ below  left = -83.5mm and -11mm of pic] {\small{0}};
			\node [ below  left = -83.5mm and -22mm of pic] {\small{50}};
			\node [ below  left = -83.5mm and -33mm of pic] {\small{100}};
			\node [ below  left = -83.5mm and -43mm of pic] {\small{150}};
			\node [ below  left = -83.5mm and -53mm of pic] {\small{200}};
			\node [ below  left = -83.5mm and -63mm of pic] {\small{250}};
			\node [ below  left = -83.5mm and -73mm of pic] {\small{300}};
			\node [ below  left = -83.5mm and -83mm of pic] {\small{350}};

			\node [ below  left = -44.5mm and -11mm of pic] {\small{0}};
			\node [ below  left = -44.5mm and -22mm of pic] {\small{50}};
			\node [ below  left = -44.5mm and -33mm of pic] {\small{100}};
			\node [ below  left = -44.5mm and -43mm of pic] {\small{150}};
			\node [ below  left = -44.5mm and -53mm of pic] {\small{200}};
			\node [ below  left = -44.5mm and -63mm of pic] {\small{250}};
			\node [ below  left = -44.5mm and -73mm of pic] {\small{300}};
			\node [ below  left = -44.5mm and -83mm of pic] {\small{350}};
			
			\node [ below  left = -5.5mm and -11mm of pic] {\small{0}};
			\node [ below  left = -5.5mm and -22mm of pic] {\small{50}};
			\node [ below  left = -5.5mm and -33mm of pic] {\small{100}};
			\node [ below  left = -5.5mm and -43mm of pic] {\small{150}};
			\node [ below  left = -5.5mm and -53mm of pic] {\small{200}};
			\node [ below  left = -5.5mm and -63mm of pic] {\small{250}};
			\node [ below  left = -5.5mm and -73mm of pic] {\small{300}};
			\node [ below  left = -5.5mm and -83mm of pic] {\small{350}};

			\node (fig) [draw=none,  minimum width = 6.27cm, minimum height = 10.63cm, opacity = 0, above left = -108.5mm and -91.2mm of pic] {};
			
			\node (r6) [draw, ultra thin, minimum width = 1mm, minimum height = 10.63cm, opacity = 0.3, fill=yellow, inner sep = 0pt, right  = -14.55mm of fig] {};
			\node (r5) [draw, ultra thin, minimum width = 1mm, minimum height = 10.63cm, opacity = 0.3, fill=yellow, inner sep = 0pt, right  = -20.35mm of fig] {};
			\node (r4) [draw, ultra thin, minimum width = 2.45mm, minimum height = 10.63cm, opacity = 0.3, fill=yellow, inner sep = 0pt, right  = -28mm of fig] {};
			\node (r3) [draw, ultra thin, minimum width = 0.75mm, minimum height = 10.63cm, opacity = 0.3, fill=yellow, inner sep = 0pt, right  = -43.3mm of fig] {};
			\node (r2) [draw, ultra thin, minimum width = 1mm, minimum height = 10.63cm, opacity = 0.3, fill=yellow, inner sep = 0pt, right  = -58.45mm of fig] {};
			\node (r1) [draw, ultra thin, minimum width = 0.75mm, minimum height = 10.63cm, opacity = 0.3, fill=yellow, inner sep = 0pt, right  = -61.35mm of fig] {};
			
			\node (A) [above left = 0mm and -2mm of r1] {$A_e$};
			\node (B) [above right = 0mm and -2mm of r2] {$B_e$};
			\node (C) [above = 0mm of r3] {$C_e$};
			\node (D) [above = 0mm of r4] {$D_e$};
			\node (E) [above = 0mm of r5] {$E_e$};
			\node (F) [above = 0mm of r6] {$F_e$};
			
			\node [below = -1mm of pic] {Time (s)};
			
			\end{tikzpicture}
			\par\end{centering}
		\protect\caption{Boost pressure and EGR rate trajectories over EUDC for baseline calibration (top and middle) and active controller (bottom). The controllers in the highlighted sections $A_e - F_e$ require further calibration.}
		\label{fig:eudc_base}
	\end{figure}

	The active controllers corresponding to the highlighted regions are identified from the bottom plot of Fig.~\ref{fig:eudc_base}. For instance, in region $B_e$, the controllers V and VI are active. Now, these controllers are calibrated according to the tuning rules from Table~\ref{tab:caltab} based on the identified types of oscillations.

	\begin{figure}
		\begin{centering}
			\begin{tikzpicture}
			\pgfkeys{/pgf/number format/.cd,fixed,precision=2}
			
			\node (pic) at (0,0)
			{\includegraphics[clip, trim = {0.7cm 1.4cm 1.3cm 1.2cm},width = 8.7cm]{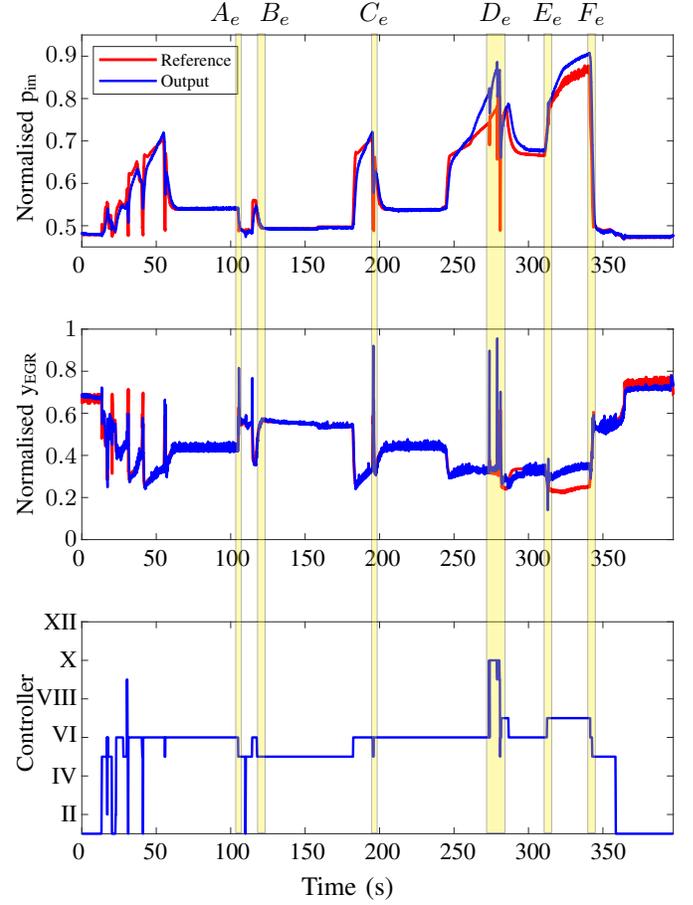}};
			
			
			\node (rect) [draw = none , fill = white, minimum width = 0.76cm, minimum height = 11cm, inner sep = 0pt, below  left = -113mm and -9mm of pic] {};
			
						\node (rect) [draw = none, fill = white, minimum width = 8.1cm, minimum height = 0.3cm, inner sep = 0pt, below  left = -83mm and -89mm of pic] {};
			
			\node (rect) [draw = none, fill = white, minimum width = 8.1cm, minimum height = 0.3cm, inner sep = 0pt, below  left = -44mm and -89mm of pic] {};
			
			\node (rect) [draw = none, fill = white, minimum width = 8.1cm, minimum height = 0.3cm, inner sep = 0pt, below  left = -5mm and -89mm of pic] {};
			
			\node [ below  left = -111.5mm and -10mm of pic] {\small{\pgfmathparse{180/\pimbar}\pgfmathprintnumber\pgfmathresult}};
			\node [ below  left = -105.5mm and -10mm of pic] {\small{\pgfmathparse{160/\pimbar}\pgfmathprintnumber\pgfmathresult}};
			\node [ below  left = -100mm and -10mm of pic] {\small{\pgfmathparse{140/\pimbar}\pgfmathprintnumber\pgfmathresult}};
			\node [ below  left = -94.5mm and -10mm of pic] {\small{\pgfmathparse{120/\pimbar}\pgfmathprintnumber\pgfmathresult}};
			\node [ below  left = -88.5mm and -10mm of pic] {\small{\pgfmathparse{100/\pimbar}\pgfmathprintnumber\pgfmathresult}};
			
			\node [ rotate = 90, below left = -108mm and 0.5mm of pic] {\small{Normalised p$_{\textrm{im}}$}};
			
			\node [ below  left = -75mm and -10mm of pic] {\small{\pgfmathparse{100/\egrbar}\pgfmathprintnumber\pgfmathresult}};
			\node [ below  left = -69.5mm and -10mm of pic] {\small{\pgfmathparse{80/\egrbar}\pgfmathprintnumber\pgfmathresult}};
			\node [ below  left = -63.5mm and -10mm of pic] {\small{\pgfmathparse{60/\egrbar}\pgfmathprintnumber\pgfmathresult}};
			\node [ below  left = -58mm and -10mm of pic] {\small{\pgfmathparse{40/\egrbar}\pgfmathprintnumber\pgfmathresult}};
			\node [ below  left = -52.5mm and -10mm of pic] {\small{\pgfmathparse{20/\egrbar}\pgfmathprintnumber\pgfmathresult}};
			\node [ below  left = -47mm and -10mm of pic] {\small{\pgfmathparse{0/\egrbar}\pgfmathprintnumber\pgfmathresult}};
			
			\node [ rotate = 90, below left = -71mm and 0.5mm of pic] { \small{Normalised {y$_{\textrm{EGR}}$}}};
			
			\node [ below  left = -36mm and -10mm of pic] {\small{XII}};
			\node [ below  left = -31mm and -10mm of pic] {\small{X}};
			\node [ below  left = -26mm and -10mm of pic] {\small{VIII}};
			\node [ below  left = -21mm and -10mm of pic] {\small{VI}};
			\node [ below  left = -15.5mm and -10mm of pic] {\small{IV}};
			\node [ below  left = -10.5mm and -10mm of pic] {\small{II}};
			
			\node [ rotate = 90, below left = -28mm and 0.5mm of pic] { \small{Controller}};


			\node [ below  left = -83.5mm and -11mm of pic] {\small{0}};
			\node [ below  left = -83.5mm and -22mm of pic] {\small{50}};
			\node [ below  left = -83.5mm and -33mm of pic] {\small{100}};
			\node [ below  left = -83.5mm and -43mm of pic] {\small{150}};
			\node [ below  left = -83.5mm and -53mm of pic] {\small{200}};
			\node [ below  left = -83.5mm and -63mm of pic] {\small{250}};
			\node [ below  left = -83.5mm and -73mm of pic] {\small{300}};
			\node [ below  left = -83.5mm and -83mm of pic] {\small{350}};

			\node [ below  left = -44.5mm and -11mm of pic] {\small{0}};
			\node [ below  left = -44.5mm and -22mm of pic] {\small{50}};
			\node [ below  left = -44.5mm and -33mm of pic] {\small{100}};
			\node [ below  left = -44.5mm and -43mm of pic] {\small{150}};
			\node [ below  left = -44.5mm and -53mm of pic] {\small{200}};
			\node [ below  left = -44.5mm and -63mm of pic] {\small{250}};
			\node [ below  left = -44.5mm and -73mm of pic] {\small{300}};
			\node [ below  left = -44.5mm and -83mm of pic] {\small{350}};
			
			\node [ below  left = -5.5mm and -11mm of pic] {\small{0}};
			\node [ below  left = -5.5mm and -22mm of pic] {\small{50}};
			\node [ below  left = -5.5mm and -33mm of pic] {\small{100}};
			\node [ below  left = -5.5mm and -43mm of pic] {\small{150}};
			\node [ below  left = -5.5mm and -53mm of pic] {\small{200}};
			\node [ below  left = -5.5mm and -63mm of pic] {\small{250}};
			\node [ below  left = -5.5mm and -73mm of pic] {\small{300}};
			\node [ below  left = -5.5mm and -83mm of pic] {\small{350}};

			\node (fig) [draw=none,  minimum width = 6.27cm, minimum height = 10.63cm, opacity = 0, above left = -108.5mm and -91.2mm of pic] {};
			
			\node (r6) [draw, ultra thin, minimum width = 1mm, minimum height = 10.63cm, opacity = 0.3, fill=yellow, inner sep = 0pt, right  = -14.55mm of fig] {};
			\node (r5) [draw, ultra thin, minimum width = 1mm, minimum height = 10.63cm, opacity = 0.3, fill=yellow, inner sep = 0pt, right  = -20.35mm of fig] {};
			\node (r4) [draw, ultra thin, minimum width = 2.45mm, minimum height = 10.63cm, opacity = 0.3, fill=yellow, inner sep = 0pt, right  = -28mm of fig] {};
			\node (r3) [draw, ultra thin, minimum width = 0.75mm, minimum height = 10.63cm, opacity = 0.3, fill=yellow, inner sep = 0pt, right  = -43.3mm of fig] {};
			\node (r2) [draw, ultra thin, minimum width = 1mm, minimum height = 10.63cm, opacity = 0.3, fill=yellow, inner sep = 0pt, right  = -58.45mm of fig] {};
			\node (r1) [draw, ultra thin, minimum width = 0.75mm, minimum height = 10.63cm, opacity = 0.3, fill=yellow, inner sep = 0pt, right  = -61.35mm of fig] {};
			
			\node (A) [above left = 0mm and -2mm of r1] {$A_e$};
			\node (B) [above right = 0mm and -2mm of r2] {$B_e$};
			\node (C) [above = 0mm of r3] {$C_e$};
			\node (D) [above = 0mm of r4] {$D_e$};
			\node (E) [above = 0mm of r5] {$E_e$};
			\node (F) [above = 0mm of r6] {$F_e$};
			
			\node [below = -1mm of pic] {Time (s)};
			
			\end{tikzpicture}
			\par\end{centering}
		\protect\caption{Boost pressure and EGR rate trajectories over EUDC for final tuning parameters (top and middle) and active controller (bottom).}
		\label{fig:eudc_final}
	\end{figure}

	The output trajectories for the final tuning parameters and the corresponding sections are shown in Fig.~\ref{fig:eudc_final}. It can be seen that the oscillations in the EGR rate response in region $B_e$ in Fig.~\ref{fig:eudc_final}, are smoothed and hence by following Algorithm~\ref{alg:drivecyclecal}, a desired transient response was achieved by appropriate calibration of the controllers V and VI using the reduced set of tuning parameters. 
	
	\begin{table}
		\begin{centering}
			\protect\caption{Problems identified in the highlighted regions of the output trajectories over EUDC and the type of oscillation.}
			\par
			\centering{}
			\begin{tabular}{>{\centering\arraybackslash} m{0.85cm} >{\centering\arraybackslash} m{3.5cm} >{\centering\arraybackslash} m{1.3cm} >{\centering\arraybackslash} m{1.3cm} }
				\hline
				Region & Identified problem in the output trajectories & Controllers & Type of oscillation \tabularnewline
				\hline
				$A_e$ & Overshoot in EGR rate & VI & $\left(0,\, 2\right)$  \tabularnewline
				$B_e$ & Oscillations in EGR rate & V, VI & $\left(0,\, 3\right)$  \tabularnewline
				$C_e$ & Overshoot in EGR rate & V & $\left(0,\, 2\right)$  \tabularnewline
				$D_e$ & Oscillatory behaviour in boost pressure and EGR rate & V, VI, X, XI & $\left(3,\, 3\right)$  \tabularnewline
				$E_e$ & Undershoot in EGR rate & VI, VII & $\left(0,\, 1\right)$  \tabularnewline
				$F_e$ & Overshoot in EGR rate & VII & $\left(0,\, 2\right)$  \tabularnewline
				\hline
			\end{tabular}\label{tab:eudc_base}
		\end{centering}
	\end{table}
	
	The improvements in the output response achieved through calibration of the selected controllers are reported in Table~\ref{tab:eudc_final}. In regions such as $D_e$ in Fig.~\ref{fig:eudc_base}, an oscillatory behaviour is observed in both output channels. These oscillations are caused due to large model mismatch while operating near a boundary, resulting in repeated switching of the controllers. By introducing more linearisation points, the model mismatch can be limited, however, additional linearisation points can increase the calibration effort. In this work, the use of 12 controllers provided a sufficient balance between enough model accuracy to provide good transient response without introducing too much switching between controllers. The resulting output responses following calibration are shown in Fig.~\ref{fig:eudc_final}.
	
	\begin{table}
		\begin{centering}
			\protect\caption{Improvements in the trajectories after calibration.}
			\par
			\centering{}
			\begin{tabular}{>{\centering\arraybackslash} m{0.85cm} >{\centering\arraybackslash} m{3.3cm} >{\centering\arraybackslash} m{3.3cm} }
				\hline
				Region & Before calibration & After calibration \tabularnewline
				\hline
				$A_e$ & Overshoot in EGR rate & Overshoot magnitude is reduced \tabularnewline
				$B_e$ & Oscillation in EGR rate & Oscillations removed \tabularnewline
				$C_e$ & Overshoot in EGR rate & No improvement \tabularnewline
				$D_e$ & Oscillatory behaviour in boost and EGR rate & Reduced oscillations in both outputs  \tabularnewline
				$E_e$ & Undershoot in EGR rate & Undershoot magnitude is reduced \tabularnewline
				$F_e$ & Overshoot in EGR rate & Overshoot is removed   \tabularnewline
				\hline
			\end{tabular}\label{tab:eudc_final}
		\end{centering}
	\end{table}

	\begin{figure}
		\begin{centering}
			\begin{tikzpicture}
			\pgfkeys{/pgf/number format/.cd,fixed,precision=2}
			\node (pic) at (0,0)
			{\includegraphics[clip, trim = {0.7cm 6.75cm 1.3cm 1.2cm},width = 8.7cm]{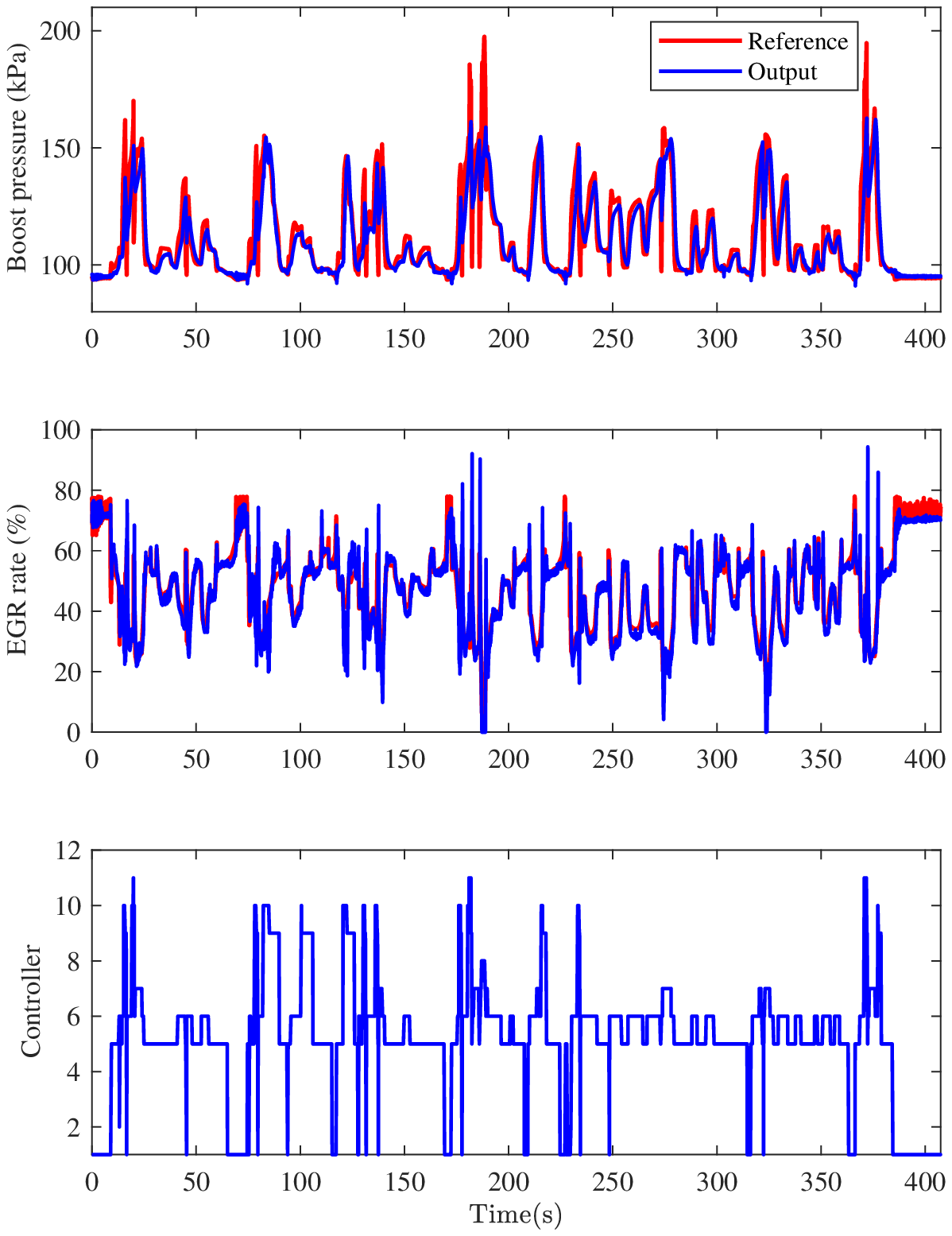}};
			
			
			\node (rect) [draw =none , fill = white, minimum width = 0.76cm, minimum height = 6.8cm, inner sep = 0pt, below  left = -73mm and -8.7mm of pic] {};
			
						\node (rect) [draw = none, fill = white, minimum width = 8.1cm, minimum height = 0.3cm, inner sep = 0pt, below  left = -44mm and -89mm of pic] {};
			
			\node (rect) [draw = none, fill = white, minimum width = 8.1cm, minimum height = 0.3cm, inner sep = 0pt, below  left = -5mm and -89mm of pic] {};
			
			\node [ below  left = -73mm and -10mm of pic] {\small{\pgfmathparse{200/\pimbar}\pgfmathprintnumber\pgfmathresult}};
			\node [ below  left = -62mm and -10mm of pic] {\footnotesize{\pgfmathparse{150/\pimbar}\pgfmathprintnumber\pgfmathresult}};
			\node [ below  left = -51mm and -10mm of pic] {\small{\pgfmathparse{100/\pimbar}\pgfmathprintnumber\pgfmathresult}};
			
			\node [ rotate = 90, below left = -71mm and 0.5mm of pic] {\small{Normalised p$_{\textrm{im}}$}};
			
			\node [ below  left = -36mm and -10mm of pic] {\small{\pgfmathparse{100/\egrbar}\pgfmathprintnumber\pgfmathresult}};
			\node [ below  left = -30.5mm and -10mm of pic] {\small{\pgfmathparse{80/\egrbar}\pgfmathprintnumber\pgfmathresult}};
			\node [ below  left = -25mm and -10mm of pic] {\small{\pgfmathparse{60/\egrbar}\pgfmathprintnumber\pgfmathresult}};
			\node [ below  left = -19mm and -10mm of pic] {\small{\pgfmathparse{40/\egrbar}\pgfmathprintnumber\pgfmathresult}};
			\node [ below  left = -14mm and -10mm of pic] {\small{\pgfmathparse{20/\egrbar}\pgfmathprintnumber\pgfmathresult}};
			\node [ below  left = -8mm and -10mm of pic] {\small{\pgfmathparse{0/\egrbar}\pgfmathprintnumber\pgfmathresult}};
			\node [ rotate = 90, below left = -32mm and 0.5mm of pic] { \small{Normalised {y$_{\textrm{EGR}}$}}};


			\node [ below  left = -44.5mm and -11mm of pic] {\small{0}};
			\node [ below  left = -44.5mm and -22mm of pic] {\small{50}};
			\node [ below  left = -44.5mm and -33mm of pic] {\small{100}};
			\node [ below  left = -44.5mm and -42mm of pic] {\small{150}};
			\node [ below  left = -44.5mm and -52mm of pic] {\small{200}};
			\node [ below  left = -44.5mm and -61.5mm of pic] {\small{250}};
			\node [ below  left = -44.5mm and -71mm of pic] {\small{300}};
			\node [ below  left = -44.5mm and -80.5mm of pic] {\small{350}};
			\node [ below  left = -44.5mm and -89.5mm of pic] {\small{400}};
			
			\node [ below  left = -5.5mm and -11mm of pic] {\small{0}};
			\node [ below  left = -5.5mm and -22mm of pic] {\small{50}};
			\node [ below  left = -5.5mm and -33mm of pic] {\small{100}};
			\node [ below  left = -5.5mm and -42mm of pic] {\small{150}};
			\node [ below  left = -5.5mm and -52mm of pic] {\small{200}};
			\node [ below  left = -5.5mm and -61.5mm of pic] {\small{250}};
			\node [ below  left = -5.5mm and -71mm of pic] {\small{300}};
			\node [ below  left = -5.5mm and -80.5mm of pic] {\small{350}};
			\node [ below  left = -5.5mm and -89.5mm of pic] {\small{400}};
			
			\node (fig)  [draw= none, minimum width = 6.27cm, minimum height = 6.73cm, opacity = 0, above left = -69.4mm and -91.2mm of pic] {};
			
			\node (r5) [draw, ultra thin, minimum width = 1.mm, minimum height = 6.71cm, opacity = 0.3, fill=yellow, inner sep = 0pt, right  = -10.6mm of fig] {};
			
			\node (r4) [draw, ultra thin, minimum width = 1.mm, minimum height = 6.71cm, opacity = 0.3, fill=yellow, inner sep = 0pt, right  = -37.2mm of fig] {};
			
			\node (r3) [draw, ultra thin, minimum width = 1.mm, minimum height = 6.71cm, opacity = 0.3, fill=yellow, inner sep = 0pt, right  = -46.0mm of fig] {};
			
			\node (r2) [draw, ultra thin, minimum width = 1.mm, minimum height = 6.71cm, opacity = 0.3, fill=yellow, inner sep = 0pt, right  = -48.0mm of fig] {};
			
			\node (r1) [draw, ultra thin, minimum width = 1.5mm, minimum height = 6.71cm, opacity = 0.3, fill=yellow, inner sep = 0pt, right  = -66.5mm of fig] {};
			
			\node (A) [above  = 0mm of r1] {$A_w$};
			\node (B) [above left = 0mm and -2mm of r2] {$B_w$};
			\node (C) [above right = 0mm and -2mm of r3] {$C_w$};
			\node (D) [above = 0mm of r4] {$D_w$};
			\node (E) [above = 0mm of r5] {$E_w$};
			
			\node [below = -1mm of pic] {Time (s)};

			\end{tikzpicture}
			\par\end{centering}
		\protect\caption{Boost pressure and EGR rate trajectories over WLTP-medium cycle for baseline calibration parameters.}
		\label{fig:wltp_base}
	\end{figure}
	
	\begin{figure}
		\begin{centering}
			\begin{tikzpicture}
			\node (pic) at (0,0)
			{\includegraphics[clip, trim = {0.7cm 6.75cm 1.3cm 1.2cm},width = 8.7cm]{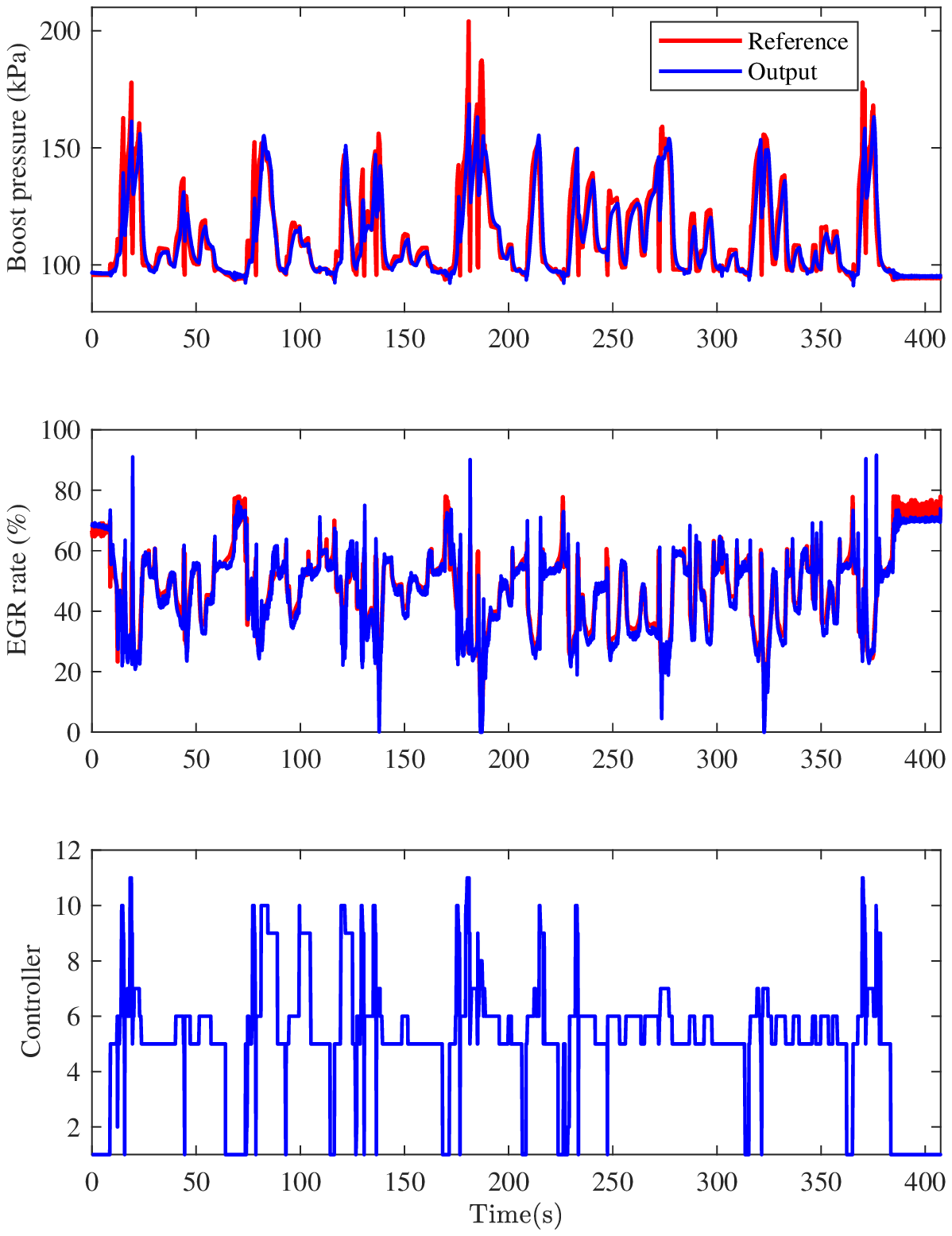}};
			
			
			\node (rect) [draw =none , fill = white, minimum width = 0.76cm, minimum height = 6.8cm, inner sep = 0pt, below  left = -73mm and -8.7mm of pic] {};
			
						\node (rect) [draw = none, fill = white, minimum width = 8.1cm, minimum height = 0.3cm, inner sep = 0pt, below  left = -44mm and -89mm of pic] {};
			
			\node (rect) [draw = none, fill = white, minimum width = 8.1cm, minimum height = 0.3cm, inner sep = 0pt, below  left = -5mm and -89mm of pic] {};
			
			\node [ below  left = -73mm and -10mm of pic] {\small{\pgfmathparse{200/\pimbar}\pgfmathprintnumber\pgfmathresult}};
			\node [ below  left = -62mm and -10mm of pic] {\footnotesize{\pgfmathparse{150/\pimbar}\pgfmathprintnumber\pgfmathresult}};
			\node [ below  left = -51mm and -10mm of pic] {\small{\pgfmathparse{100/\pimbar}\pgfmathprintnumber\pgfmathresult}};
			
			\node [ rotate = 90, below left = -71mm and 0.5mm of pic] {\small{Normalised p$_{\textrm{im}}$}};
			
			\node [ below  left = -36mm and -10mm of pic] {\small{\pgfmathparse{100/\egrbar}\pgfmathprintnumber\pgfmathresult}};
			\node [ below  left = -30.5mm and -10mm of pic] {\small{\pgfmathparse{80/\egrbar}\pgfmathprintnumber\pgfmathresult}};
			\node [ below  left = -25mm and -10mm of pic] {\small{\pgfmathparse{60/\egrbar}\pgfmathprintnumber\pgfmathresult}};
			\node [ below  left = -19mm and -10mm of pic] {\small{\pgfmathparse{40/\egrbar}\pgfmathprintnumber\pgfmathresult}};
			\node [ below  left = -14mm and -10mm of pic] {\small{\pgfmathparse{20/\egrbar}\pgfmathprintnumber\pgfmathresult}};
			\node [ below  left = -8mm and -10mm of pic] {\small{\pgfmathparse{0/\egrbar}\pgfmathprintnumber\pgfmathresult}};
			\node [ rotate = 90, below left = -32mm and 0.5mm of pic] { \small{Normalised {y$_{\textrm{EGR}}$}}};


			\node [ below  left = -44.5mm and -11mm of pic] {\small{0}};
			\node [ below  left = -44.5mm and -22mm of pic] {\small{50}};
			\node [ below  left = -44.5mm and -33mm of pic] {\small{100}};
			\node [ below  left = -44.5mm and -42mm of pic] {\small{150}};
			\node [ below  left = -44.5mm and -52mm of pic] {\small{200}};
			\node [ below  left = -44.5mm and -61.5mm of pic] {\small{250}};
			\node [ below  left = -44.5mm and -71mm of pic] {\small{300}};
			\node [ below  left = -44.5mm and -80.5mm of pic] {\small{350}};
			\node [ below  left = -44.5mm and -89.5mm of pic] {\small{400}};
			
			\node [ below  left = -5.5mm and -11mm of pic] {\small{0}};
			\node [ below  left = -5.5mm and -22mm of pic] {\small{50}};
			\node [ below  left = -5.5mm and -33mm of pic] {\small{100}};
			\node [ below  left = -5.5mm and -42mm of pic] {\small{150}};
			\node [ below  left = -5.5mm and -52mm of pic] {\small{200}};
			\node [ below  left = -5.5mm and -61.5mm of pic] {\small{250}};
			\node [ below  left = -5.5mm and -71mm of pic] {\small{300}};
			\node [ below  left = -5.5mm and -80.5mm of pic] {\small{350}};
			\node [ below  left = -5.5mm and -89.5mm of pic] {\small{400}};
			
			\node (fig)  [draw= none, minimum width = 6.27cm, minimum height = 6.73cm, opacity = 0, above left = -69.4mm and -91.2mm of pic] {};
			
			\node (r5) [draw, ultra thin, minimum width = 1.mm, minimum height = 6.71cm, opacity = 0.3, fill=yellow, inner sep = 0pt, right  = -10.6mm of fig] {};
			
			\node (r4) [draw, ultra thin, minimum width = 1.mm, minimum height = 6.71cm, opacity = 0.3, fill=yellow, inner sep = 0pt, right  = -37.2mm of fig] {};
			
			\node (r3) [draw, ultra thin, minimum width = 1.mm, minimum height = 6.71cm, opacity = 0.3, fill=yellow, inner sep = 0pt, right  = -46.0mm of fig] {};
			
			\node (r2) [draw, ultra thin, minimum width = 1.mm, minimum height = 6.71cm, opacity = 0.3, fill=yellow, inner sep = 0pt, right  = -48.0mm of fig] {};
			
			\node (r1) [draw, ultra thin, minimum width = 1.5mm, minimum height = 6.71cm, opacity = 0.3, fill=yellow, inner sep = 0pt, right  = -66.5mm of fig] {};
			
			\node (A) [above  = 0mm of r1] {$A_w$};
			\node (B) [above left = 0mm and -2mm of r2] {$B_w$};
			\node (C) [above right = 0mm and -2mm of r3] {$C_w$};
			\node (D) [above = 0mm of r4] {$D_w$};
			\node (E) [above = 0mm of r5] {$E_w$};
			
			\node [below = -1mm of pic] {Time (s)};
			
			\end{tikzpicture}
			\par\end{centering}
		\protect\caption{Boost pressure and EGR rate trajectories over WLTP-medium cycle for final calibration parameters.}
		\label{fig:wltp_final}
	\end{figure}
	
	\subsubsection{Calibration over WLTP-Medium Phase}
	The proposed diesel airpath control framework is tested when the engine is running over the medium phase of WLTP. The closed-loop response obtained with the baseline parameters $\left(\tau_{\textrm{boost}}^{g} = 0.5,\, \tau_{\textrm{EGR}}^{g} = 0.5,\, w^{g}= 0.5,\, \epsilon_{\textrm{boost}}^{g} = 0 \textrm{ and } \epsilon_{\textrm{EGR}}^{g} = 0\right)$ is shown in Fig.~\ref{fig:wltp_base}. Five regions are identified to improve the transient response in one or both output channels. Oscillations are observed in both output channels in region $A_w$ while overshoots and an undershooting behaviour in the EGR rate channel are observed in regions $B_w$, $C_w$, $E_w$ and $D_w$, respectively. 
	
	The output trajectories for the final tuning parameters over WLTP-medium phase are shown in Fig.~\ref{fig:wltp_final}. The transient responses achieved through calibration of selected local controllers in the two regions $D_w$ and $E_w$, have a reduction in the magnitudes of the undershoot and overshoot, respectively. On the other hand, significant improvements in the output responses are observed in region $A_w$, where the oscillations are smoothed in boost pressure and reduced in the EGR rate response, and in regions $B_w$ and $C_w$, where the overshoots are removed completely. These improvements in the output responses were achieved with a low degree of calibration effort utilising the proposed controller structure and calibration rules.

	\section{Conclusions}
	\label{sec:con}
	
	In this paper, a calibration-friendly model predictive controller has been proposed for diesel engines and its calibration efficacy has been experimentally demonstrated. The proposed MPC cost function parameterisation and the controller structure have facilitated a significant reduction in the number of effective tuning parameters compared to the conventional MPC. In addition, the new set of tuning parameters has an intuitive correlation with the output transient response which helps with efficient engine calibration.
	
	A switched LTI-MPC architecture with multiple linear controllers has been proposed to handle the transient operation of the engine. The maximal disturbance set that can be handled by each local controller and the corresponding constraint tightening margins have been obtained as a solution to an offline sequential convex program. The tightening margins are then used in the online MPC optimisation to guarantee constraint adherence in the presence of disturbances originating from the maximal disturbance set. 
	
	The calibration efficacy of the controller architecture, based on the developed calibration rules,  has been experimentally demonstrated at a steady state condition for fuelling step and over EUDC and the medium phase of WLTP cycles. The proposed calibration methodology provides a systematic approach in which the under-performing controllers are identified and calibrated using the reduced set of tuning parameters, resulting in a good tracking performance in both output channels over both drive cycles. The improvements achieved in the output transient responses by tuning a small number of parameters demonstrate the effectiveness of the proposed MPC formulation in reducing the calibration effort and aiding in rapid calibration.
	
	Future research can test the performance of the proposed approach with severely ill-conditioned families of linear models, where adjacent grid points have zero-gain in some input-output channels or a gain sign change. This may provide insight for developing extensions to the proposed tuning rules to reduce or remove oscillatory responses as well as minimise the time actuators spend in saturation.
	
	\section*{Acknowledgements}
	The authors would like to thank the engineering staff at Toyota's Higashi-Fuji Technical Center in Susono, Japan, for assisting with the experiments.
	

	\bibliographystyle{IEEEtran}
	\bibliography{ref}
	

\vskip 0pt plus -1fil
\begin{IEEEbiography}
	[{\includegraphics[width=1in,height=1.25in,clip,keepaspectratio]{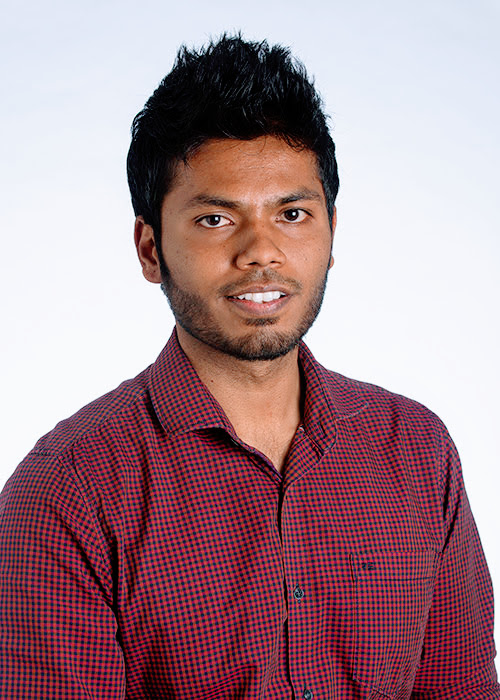}}]{Gokul S. SanKar} pursued his Ph.D. at the University of Melbourne, Parkville, VIC, Australia, and received his M.S. from Indian Institute of Technology (Madras), India, in 2013 and Bachelor’s degree in Electronics $\&$ Instrumentation Engineering, from Anna University, India, in 2010.
	He is currently a Research Fellow at the University of Michigan, Ann Arbor, MI, USA. His research interests include model predictive control with applications to automotive systems, autonomous vehicles, and robust control algorithms.
	He was a member of exclusive Melbourne-India postgraduate program (MIPP) cohort at the University of Melbourne during his doctoral studies.
\end{IEEEbiography}
	\vskip 0pt plus -1fil
	
	\begin{IEEEbiography}
		[{\includegraphics[width=1in,height=1.25in,clip,keepaspectratio]{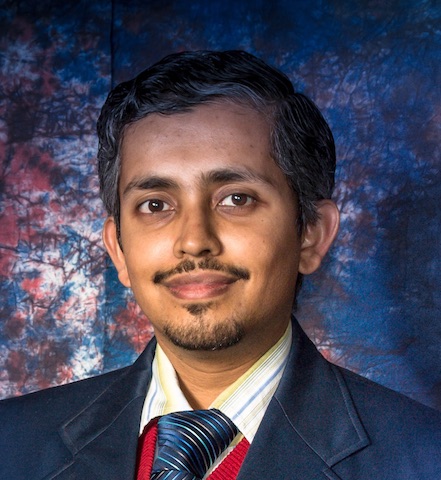}}] {Rohan C. Shekhar} received the B.E. degree (with honours) in Mechatronic Engineering from The University of Queensland in 2006. In 2008, he was awarded the Sir Robert Menzies Memorial Scholarship in Engineering and an Honorary Poynton Cambridge Australia Fellowship for undertaking doctoral studies at the University of Cambridge, receiving the Ph.D. degree in 2012. From 2012-2017, he conducted postdoctoral research at the University of Melbourne.  Rohan is currently an Assistant Professor with the Minerva Schools at KGI while maintaining an honorary appointment at Melbourne.  His research centres on robust model predictive control, with applications to automotive systems, autonomous vehicles, mining technology and robotics.
	\end{IEEEbiography}
	\vskip 0pt plus -1fil

		\begin{IEEEbiography}
		[{\includegraphics[width=1in,height=1.25in,clip,keepaspectratio]{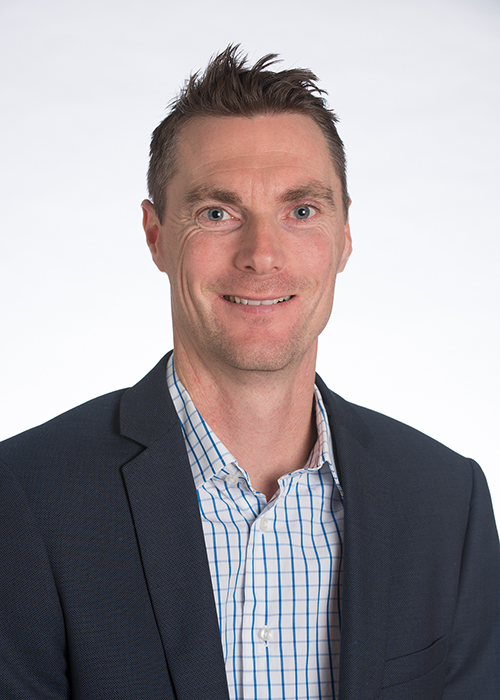}}] {Chris Manzie} is a Professor and Head of Department of Electrical and Electronic Engineering at the University of Melbourne, and also the Director of the cross-disciplinary Melbourne Information, Decision and Autonomous Systems (MIDAS) Laboratory. Over the period 2003-2016, he was an academic in the Department of Mechanical Engineering, with responsibilities including Assistant Dean with the portfolio of Research Training (2011-2017), and Mechatronics Program Director (2009-2016). Professor Manzie has had visiting positions with the University of California, San Diego and IFP Energies Nouvelles, Rueil Malmaison.  His research interests are in model-based and model-free control and optimisation, with applications in a range of areas including systems related to autonomous systems, energy, transportation and mechatronics. He is currently an Associate Editor for IEEE Transactions on Control Systems Technology and Elsevier Mechatronics and is a past Associate Editor for Control Engineering Practice and IEEE/ASME Transactions on Mechatronics.
	\end{IEEEbiography}
	\vskip 0pt plus -1fil

		\begin{IEEEbiography}
		[{\includegraphics[width=1in,height=1.25in,clip,keepaspectratio]{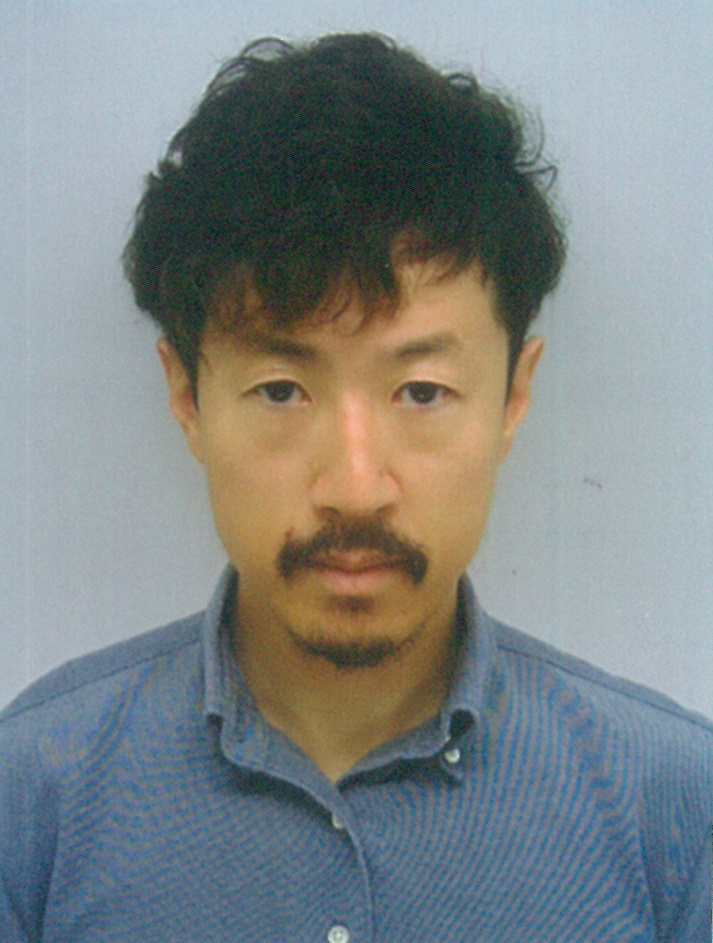}}] {Takeshi Sano} received B.Eng. degree in mechanical engineering and M.Eng. (Master of Engineering) in combustion engineering from Tokyo Denki University, Japan in 2005, 2007. Since 2007, he has been with the Advanced Powertrain Management System Development Division of Toyota Motor Corporation, Japan. He is working on development of engine and hybrid powertrain system.
	\end{IEEEbiography}
	\vskip 0pt plus -1fil

		\begin{IEEEbiography}
		[{\includegraphics[width=1in,height=1.25in,clip,keepaspectratio]{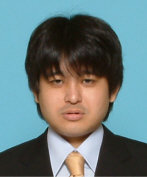}}] {Hayato Nakada} received B.Eng. degree in informatics and mathematical science and M.Inf. (Master of Informatics) and Ph.D. degrees in applied mathematics and physics from Kyoto University, Japan in 2000, 2002 and 2005, respectively, specializing in control theory. Since 2005, he has been with the Advanced Powertrain Management System Development Division of Toyota Motor Corporation, Japan. He is employed as a Project Manager, working on development of advanced controllers and models for a range of powertrain projects.
	\end{IEEEbiography}
	\vskip 0pt plus -1fil
	
\end{document}